\documentclass[12pt, draftclsnofoot, onecolumn]{IEEEtran}
\usepackage[colorlinks,
linkcolor=red,
anchorcolor=green,
citecolor=green,
]{hyperref} % Add hyperlinks
\usepackage{amsfonts,amssymb,amsmath,array}
\usepackage{latexsym}
\usepackage{CJK}
\usepackage{cite}
\usepackage{bm}
\usepackage{color, soul}
\usepackage{colortbl}

\usepackage{graphicx}
\usepackage{epstopdf}
\usepackage{epsfig}
\usepackage{subfigure} % 图片并排显示
\usepackage{framed}
\usepackage{verbatim}  % \begin{comment} ... \end{comment}
\usepackage{color,soul}
\definecolor{aliceblue}{rgb}{0.94, 0.97, 1.0}
\definecolor{blizzardblue}{rgb}{0.67, 0.9, 0.93}
\definecolor{antiquebrass}{rgb}{0.8, 0.58, 0.46}
\definecolor{beaublue}{rgb}{0.74, 0.83, 0.9}
\sethlcolor{blizzardblue}
\usepackage{mathrsfs}
\usepackage{algorithm} %format of the algorithm
\usepackage{algorithmic} %format of the algorithm
\usepackage{booktabs}
\usepackage{textcomp}
\usepackage{multirow}
\usepackage{lettrine}
\usepackage[scaled=0.92]{helvet}

\usepackage{algorithm}
\usepackage{algorithmic}
\usepackage{subfigure} % 图片并排显示
\usepackage{framed}
\usepackage{verbatim}  % \begin{comment} ... \end{comment}

\usepackage{diagbox} %表格单元格内斜线
\usepackage{tcolorbox}

\usepackage{amsthm} %thm / proof
\theoremstyle{remark} %标题斜体，内容罗马体

\newtheorem{thm}{Theorem}%[section]
\newtheorem{cor}{Corollary}
\newtheorem{lem}{Lemma}
\newtheorem{rem}{Remark}
\usepackage{accents} %长的上横线
\makeatletter
\def\widebar{\accentset{{\cc@style\underline{\mskip10mu}}}}
\def\Widebar{\accentset{{\cc@style\underline{\mskip13mu}}}}
\begin{document}
%\title{How Many UAVs Are Needed: Optimal 3D Aerial Base-stations Placement, Bayesian Learning and Empirical Traffic Patterns}
%\title{Towards Bayesian Learning and Aerial Base-Stations: 3D UAV Placement and Migration}
%\title{Learning the Traffic Patterns in UAV-BS Systems: 3D Placement and UAV Migration}
%\title{Toward Empirical Traffic Patterns in UAV-BS System: 3D Placement and Migration}
%\title{3D Placement and Migration of UAV-BSs: From Machine Learning to Information Theory}
%\title{3D Placement and Migration of UAV-BSs: From Machine Learning to Communication}
%\title{Traffic Density Driven Placement of UAV-BSs: From Machine Learning to Information Theory}
%\title{When Information Theory Hits Machine Learning: An Example on Placement of UAV-BSs}
%\title{Pattern Formation Driven Placement of UAV-BSs: From Machine Learning to Information Theory}
\title{Storage Space Allocation Strategy for Digital Data with Message Importance}
%\title{Pattern Formation on Wireless Communications: Take the Placement of UAV-BSs for Example}
%\title{When Wireless Communication Hits Machine Learning: An Example on Placement of UAV-BSs}
%\title{Traffic Pattern Driven Placement of UAV-BSs: When Information Theory Hits Machine Learning:}
%\title{Data Driven Placement of UAV-BSs: When Information Theory Hits Machine Learning}
%\title{Mobile UAV-BSs for Energy-Efficient IoT Networks: From Machine Learning to Information Theory}
%\title{3D Placement and Migration of UAV-BSs: Distributively Learning is Asymptotically Optimal}
%\title{3D Placement and Migration of UAV-BSs: A Machine Learning Based Approach}
%\title{3D Placement and Migration of UAV-BSs: ML, OR and IT}
%\title{Toward Traffic Patterns in Aerial Base-Station Coverage Systems: Bayesian Learning, 3D Placement and Migration}
%\title{Toward Traffic Patterns in UAV Coverage Systems: Bayesian Learning and 3D Placement}
%\title{How Many UAVs Are Needed in 3D Aerial Base-stations Placement: Bayesian Learning and Empirical Traffic Patterns}

\author{Shanyun~Liu, Rui~She, Zheqi~Zhu, Pingyi~Fan,~\IEEEmembership{Senior Member, ~IEEE}\\\

\thanks{
Shanyun~Liu, Rui~She, Zheqi~Zhu and Pingyi Fan are with Beijing National Research Center for Information Science and Technology (BNRist) and the Department of Electronic Engineering,  Tsinghua University,  Beijing,  P.R. China, 100084. e-mail: \{liushany16,~sher15,~zhuzq18@mails.tsinghua.edu.cn, fpy@tsinghua.edu.cn.\}. %Pingyi Fan is the corresponding author, Tel.: +86-010-6279-6973.
}
%\thanks{This work was supported by the National Natural Science Foundation of China (NSFC) No. 61771283 and Beijing Natural Science Foundation (4202030).}
}

\maketitle
\begin{abstract}
This paper mainly focuses on the problem of lossy compression storage from the perspective of message importance when the reconstructed data pursues the least distortion within limited total storage size.
For this purpose, we transform this problem to an optimization by means of the importance-weighted reconstruction error in data reconstruction.
Based on it, this paper puts forward an optimal allocation strategy in the storage of digital data by a kind of restrictive water-filling.
That is, it is a high efficient adaptive compression strategy since it can make rational use of all the storage space.
It also characterizes the trade-off between the relative weighted reconstruction error and the available storage size.
Furthermore, this paper also presents that both the users' preferences and the special characteristic of data distribution can trigger the small-probability event scenarios where only a fraction of data can cover the vast majority of users' interests.
Whether it is for one of the reasons above, the data with highly clustered message importance is beneficial to compression storage.
 In contrast, the data with uniform information distribution is incompressible, which is consistent with that in information theory.
%with the increase of message importance loss because there are growth region and saturation region for the maximum received entropy rate
\end{abstract}


\IEEEpeerreviewmaketitle
\section{Introduction}
As growing mobile devices such as Internet of things (IoT) devices or smartphones are utilized, the contradiction between limited storage space and sharply increasing data deluge becomes increasingly serious in the era of big data \cite{chen2014big,cai2016iot}. 
This exceedingly massive data makes the conventional data storage mechanisms inadequate within a tolerable time, and therefore the data storage is one of the major challenges in big data \cite{hu2014toward}.
Note that, storing all the data becomes more and more dispensable nowadays, and it is also not conducive to reduce data transmission cost \cite{dong2017content,park2018data}.
In fact, the data compression storage is widely adopted in many applications, such as IoT \cite{cai2016iot}, industrial data platform \cite{geng2019big}, bioinformatics \cite{nalbantoglu2010data}, wireless networking \cite{cao2017towards}.
Thus, the research on data compression storage becomes increasingly paramount and compelling nowadays.

In the conventional source coding, data compression is gotten by removing the data redundancy, where short descriptions are assigned to most frequent class \cite{shannon2001mathematical}.
Based on it, the tight bounds for lossless data compression is given. In order to further increase the compression rate, we need to use more information.
A quintessential example is that we can do source coding with side information \cite{oohama2018exponential}. 
Another possible solution is to compress the data with quiet a few losses first and then reconstruct them with acceptable distortion \cite{pourkamali2017preconditioned,aguerri2016lossy,cui2012distributed}.
In addition, the adaptive compression is adopted extensively.
For example, Ref. \cite{ukil2015adaptive} proposed an adaptive compression scheme in IoT systems, and backlog-adaptive source coding system in age of information is discussed in Ref. \cite{zhang2017backlog}.
In fact, most previous compression methods achieved the target of compression by means of contextual data or leveraging data transformation techniques \cite{dong2017content}. Instead of compressing data based on removing data redundancy or data correlation, as an alternative, this paper will realize this goal by reallocating storage space with taking the importance as the weight in the weighted reconstruction error to minimize the difference between the raw data and the compressed data when used by people. 

Generally, users prefer to care about the crucial part of data that attracts their attentions rather than the whole data itself.
Moreover, different errors may bring different costs in many real-world applications \cite{elkan2001foundations,zhou2006trainging,lomax2013survey,masnick1967on}. 
To be specific, the distortion in the data that users care about may be catastrophic while the loss of the data that is insignificance for users is usually inessential.
Therefore, we can achieve data compression by storing a fraction of data which preserves as much information as possible regarding the data that users care about \cite{liu2017non,tegmark2020pareto}.
This paper also employs this strategy.
However, there are subtle but critical differences between the compression storage strategy proposed in this paper with those in Ref. \cite{liu2017non,tegmark2020pareto}.
In fact, Ref. \cite{tegmark2020pareto} focused on Pareto-optimal data compression, which presents the trade-off between retained
entropy and class information. However, this paper puts forward optimal compression storage strategy for digital data from the viewpoint of message importance, and it gives the trade-off between the weighted reconstruction error and the available storage size. 
Besides, the compression method based on message importance was preliminarily discussed in Ref. \cite{liu2017non} to solve the big data storage problem in wireless communications, while this paper will desire to discuss the optimal storage space allocation strategy with limited storage space in general cases based on message importance. Moreover, the constraints are also different. That is, the available storage size is limited in this paper, while the total code length of all the events is given in Ref. \cite{liu2017non}

Much of the research in the last decade suggested that the study from the perspective of message importance is rewarding to obtain new findings \cite{Ivanchev2016information,Kawanaka2017information,sun2017unequal}.
Thus, there may be effective performance improvement in storage system with taking message importance into account.
For example, Ref. \cite{li2018learning} discussed lossy image compression method with the aid of a content-weighted importance map.
Since that any quantity can be seen as the importance if it agrees with the intuitive characterization of the user's subjective concern degree of data, the cost in data reconstruction for specific user preferences is regarded as the importance in this paper, which will be used as the weight in weighted reconstruction error. 

Since we desire to gain data compression by keeping only a small portion of important data and abandoning less important data, this paper mainly focuses on the case where only a fraction of data take up the vast majority of the users' interests. Actually, this type of scenario is not rare in big data.
A quintessential example should be cited that the minority subset detection is overwhelmingly paramount in intrusion detection \cite{xiaofeng2017research,li2017application}.
Moreover, this phenomenon is also exceedingly typical in financial crime detection systems for the fact that only a few illicit identities catch our eyes to prevent financial frauds \cite{beasley2000fraudulent}.
Actually, when a certain degree of information loss can be acceptable, people prefer to take high-probability events as granted and abandon them to maximize the compressibility. This cases are referred to as {\textit{small-probability event scenarios}} in this paper.
In order to depict the message importance in small-probability event scenarios, message importance measure (MIM) was proposed in Ref. \cite{fan2016message}.
Furthermore, MIM is fairly effective in many applications in big data, such as IoT \cite{she2019importance}, mobile edge computing \cite{she2018recognizing}.
Besides, Ref. \cite{liu2019matching} expanded MIM to the general case, and it presented that MIM can be adopted as a special weight in designing the recommendation system. Thus, this paper will illuminate the properties of this new compression strategy with taking MIM as the importance weight.

In this paper, we firstly propose a particular storage space allocation strategy for digital data on the best effort in minimizing the importance-weighted reconstruction error when the total available storage size is provided.
For digital data, we formulate this problem as an optimization problem, and present the optimal storage strategy by means of a kind of restrictive water-filling.
For given available storage size, the storage size is mainly determined by the values of message importance and probability distribution of event class in data sequence.
In fact, this optimal allocation strategy adaptively prefers to provide more storage size for crucial data classes in order to make rational use of resources, which is in accord with the cognitive mechanism of human beings.
%In fact, this strategy is a high efficiency adaptive storage size allocation strategy for the fact that it can make rational use of all the storage space.

Afterwards, we focus on the properties of this optimal storage space allocation strategy when the importance weights are characterized by MIM.
It is noted that there is a trade-off between the relative weighted reconstruction error (RWRE) and the available storage size.
The constraints on the performance of this storage system are true, and they depend on the importance coefficient and probability distribution of events classes. 
On the one hand, the RWRE increases with increasing of the absolute value of importance coefficient for the fact that the overwhelming majority of important information will gather in a fraction of data as the importance coefficient increases to negative/positive infinity, which suggests the influence of users' preferences.
On the other hand, the compression performance is also affected by probability distribution of event classes. In fact, the more closely the probability distribution matches the requirement of the small-probability event scenarios, the more effective this compression strategy becomes.
Besides, it is also obtained that the uniform distribution is incompressible, which satisfies the conclusion in information theory \cite{Elements}.

The main contributions of this paper can be summarized as follows.
(1) This paper proposes a new digital data compression strategy with taking message importance into account, which can help improve the design of big data storage system.
(2) We illuminate the properties of this new method, which shows that there is a trade-off between the RWRE and the available storage size. 
(3) We find that the data with highly clustered message importance is beneficial to compression storage, while the data with uniform information distribution is incompressible.

The rest of this paper is organized as follows. The system model is introduced in Section \ref{sec: system model}, including the definition of weighted reconstruction error, distortion measure, problem formulation.
In Section \ref{sec: optimal problem}, we solve the problem of optimal storage space allocation in three kinds of system models and give the solutions.
The properties of this optimal storage space allocation strategy based on MIM are fully discussed in Section \ref{sec: MIM}. The effects of the importance coefficient and the probability of event classes on RWRE are also focused on in this section.
Section \ref{sec: NMIM} illuminate the properties of this optimal storage strategy when the importance weight is characterized by Non-parametric MIM. 
The numerical results are shown and discussed in Section \ref{Numerical Results}, which verifies the validity of proposed results in this paper.
Finally, we give the conclusion in Section \ref{sec: conclusion}. Besides, the main notations in this paper are listed in Table~\ref{tab:notation}.
\begin{table}[H]
\centering
    \caption{Notations.}\label{tab:notation}
\begin{tabular}{cl}
\toprule
\textbf{Notation} & \textbf{Description} \\
\midrule
${\emph{\textbf{x}}}=x_1,x_2,...,x_K$  & The sequence of raw data\\
$\hat{\emph{\textbf{{x}}}}=\hat{x}_1,\hat{x}_2,...,\hat{x}_k,...,\hat{x}_K$  & The sequence of compressed data\\
$S_x$ & The storage size of $x$\\
$D_f(S_{x1},S_{x2})$ & The distortion measure function between $S_{x1}$ and $S_{x2}$ in data reconstruction\\
$n$  &  The number of event classes\\
$\{a_1,a_2,...,a_n\}$  &  The alphabet of raw data\\
$\{\hat{a}_1,\hat{a}_2,...,\hat{a}_n\}$  &  The alphabet of compressed data\\
${\emph{\textbf{W}}}=\{W_1,W_2,...,W_n\}$  &  The importance weight \\
${\emph{\textbf{P}}}=\{p_1,p_2,...,p_n\}$  &   The probability distribution of data class\\
$D({\emph{\textbf{x}}},{\emph{\textbf{W}}})$  &  The weighted reconstruction error   \\
$D_r({\emph{\textbf{x}}},{\emph{\textbf{W}}}),D_r(\emph{\textbf{W}},\emph{\textbf{L}},\emph{\textbf{l}})$  &  The relative weighted reconstruction error   \\
${\emph{\textbf{L}}}=L_1,L_2,...,L_n$  & The storage size of raw data\\
${\emph{\textbf{l}}}=l_1,l_2,...,l_n$  & The storage size of compressed data\\
$l_i^*$   &  The round optimal storage size of the data belonging to the $i$-th class \\
$T$    & The maximum available storage size\\
$\varpi$  & The importance coefficient \\
$\gamma_p$  & $\gamma_p=\sum_{i=1}^n p_i^2$\\
$\alpha_1$, $\alpha_2$ & $\alpha_1=\arg \mathop {\min}\nolimits_i p_i$ and $\alpha_2=\arg \mathop {\max}\nolimits_i p_i$  \\
% & $\alpha_2=\arg \mathop {\max}\nolimits_i p_i$ and assume $p_{\alpha_2}> p_i$ for $i \ne \alpha_2$  \\
$L(\varpi,\emph{\textbf{p}})$ & The message importance measure, which is given by $L(\varpi,\emph{\textbf{p}})=\ln \sum\nolimits_{i=1}^n {p_i e^{\varpi (1-p_i)}}$  \\
$\Delta$   & The actual compressed storage size, which is given by $\Delta=L-T$   \\
$\Delta^*(\delta)$   & The maximum available compressed storage size for giving upper bound of the RWRE $\delta$  \\
$\mathcal{L}({\textbf{\emph{P}}})$ &  The non-parametric message importance, which is given by $\mathcal{L}({\textbf{\emph{P}}})=\ln \sum\nolimits_{i=1}^n p_i e^{(1-p_i)/p_i} $  \\
%\midrule
\bottomrule
\end{tabular}
\end{table}

\section{System Model}\label{sec: system model}
This section introduces the system model, including the definition of weighted reconstruction error, distortion measure, in order to illustrate how we formulate the lossy compression problem as an optimization problem for digital data based on message importance.

\subsection{Modeling Weighted Reconstruction Error Based on Importance}
We consider a storage system which stores $K$ pieces of data as shown in Figure \ref{fig:model}. Let ${\emph{\textbf{x}}}=x_1,x_2,...,x_k,...,x_K$ be the sequence of raw data, and each data $x_k$ needs to take up storage space with size of $Sx_k$ if this data can be recovered without any distortion. After storing, the compressed data sequence is $\hat{x}_1,\hat{x}_2,...,\hat{x}_k,...,\hat{x}_K$, and the compressed data $\hat{x}_k$ takes up storage space with size of $S\hat{x}_k$ in practice for $1 \le k \le K$. Furthermore, we use the notation $W_k$ to denote the cost of the error when users use the reconstructed data. Namely, $W_k$ is denoted as the importance weight of data $x_k$ for specific user preferences. Therefore, the weighted reconstruction error is given by
\begin{flalign}\label{equ:Def_D}
D({\emph{\textbf{x}}},{\emph{\textbf{W}}})=\sum\nolimits_{k=1}^K W_k D_f(Sx_k,S{\hat{x}}_k),
\end{flalign}
where $D_f(Sx_k,S{\hat{x}}_k)$ characterizes the distortion between the raw data and the compressed data in data reconstruction.

\begin{figure}[H]
  \centerline{\includegraphics[width=15.0cm]{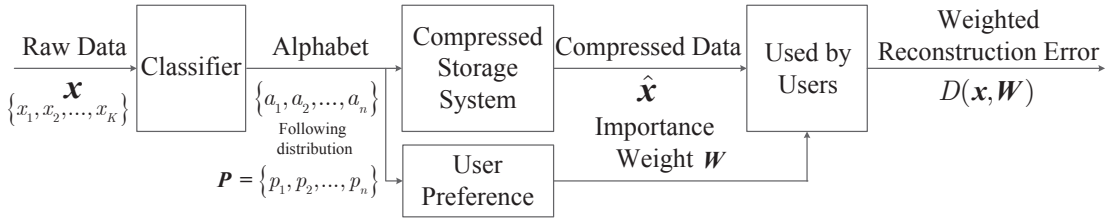}}
  %\centerline{(c) Result 3}
  \caption{Pictorial representation of the system model.}\label{fig:model}
\end{figure}

Consider the situation where the data is stored by its category for easier retrieval, which can make the recommendation system based on it more effective \cite{liu2019matching}. 
Since that data classification becomes increasingly convenient and accurate nowadays due to the rapid development of machine learning \cite{aggarwal2014data,Salvador2019compressed}, this paper assumes that the event class can be easily detected and known in storage system.
Moreover, assume the data which belongs to the same class has the same importance-weight and occupies the same storage size. Hence, $\emph{\textbf{x}}$ can be seen as a sequence of $K$ symbols from an alphabet $\{a_1,a_2,...,a_n\}$ where $a_i$ represent event class $i$. In this case, the weighted reconstruction error based on importance is formulated as
\begin{flalign}\label{equ:Def_D1}
D({\emph{\textbf{x}}},{\emph{\textbf{W}}})&=\sum\nolimits_{i=1}^n  \frac{N(a_i|{\emph{\textbf{x}} })}{K}  W_i D_f(Sa_i,S{\hat{a}}_i)   \\
&=\sum\nolimits_{i=1}^n  p_i  W_i D_f(Sa_i,S{\hat{a}}_i)     \tag{\theequation a}\label{equ:Def_D1 a},
\end{flalign}
where $N(a_i|{\emph{\textbf{x}}})$ is the number of times the  $i$-class occurs in the sequence $\emph{\textbf{x}}$. Let $p_i={N(a_i|{\emph{\textbf{x}} })}/{K}$ to denote the probability distribution of event class $i$ in data sequence $\emph{\textbf{x}}$.

\subsection{Modeling Distortion between the Raw Data and the Compressed Data}
In general, the data storage system is lack of storage space when faced with super-large scale of data. If there is limited storage resources which can be assigned to data, the optimization of storage resource allocation will be indispensable. To frame the problem appropriately, it is imperative to characterize the distortion between the raw data and the compressed data with specified storage size. Usually, there is no universal characterization of this distortion measure, especially in speech coding and image coding \cite{Elements}. In order to facilitate the analysis and design, this paper will discuss the following special case. 

We assume that the data is digital. The description of the raw data $a_i$ requires $L_i$ bits, and $a_i=\sum\nolimits_{j=0}^{L_i-1} b_j\times r^{j}$ where $r$ is radix ($r>1$). In particular, $L_i$ will approach the infinite number if $a_i$ is arbitrary real number. In the storage system of this paper, there is only $l_i$ bits assigned to it. For convenience, the smaller $L_i-l_i$ numbers is discarded, and they are random numbers in actual system. Thus, the compressed data is $\hat{a}_i=\sum\nolimits_{j=L_i-l_i}^{L_i-1} b_j\times r^{j}+\sum\nolimits_{j=0}^{L_i-l_i-1} b^*_j\times r^{j}$ where $b^*_j$ is a random number in $\{0,...,r-1\}$. The absolute error is $\vert a_i-\hat{a}_i \vert$, which meets
\begin{flalign}\label{equ:Df}
\vert a_i-\hat{a}_i \vert \le r^{L_i-l_i}-1.
\end{flalign}
When $l_i=0$, which means there is no information stored, the absolute error reaches the maximum and it is $\vert a_i-\hat{a}_i \vert=r^{L_i}-1$. This paper defines the relative error which is normalized by the maximum absolute error as the distortion measure, which is given by
\begin{flalign}\label{equ:}
D_f(Sa_i,S{\hat{a}}_i)=D_f(L_i,l_i)=\frac{r^{L_i-l_i}-1}{r^{L_i}-1}.
\end{flalign}
In particular, we obtain $D_f(L_i,L_i)=0$ and $D_f(L_i,0)=1$.
Moreover, it is easy to check that $0 \le D_f(L_i,l_i) \le 1$ and $D_f(L_i,l_i)$ decreases with the increasing of $l_i$.

To simplify the comparisons under different conditions , the weighted reconstruction error is also normalized to the RWRE. Then the RWRE is given by
\begin{flalign}\label{equ:Def_RWRE}
D_r({\emph{\textbf{x}}},{\emph{\textbf{W}}})=D_r(\emph{\textbf{W}},\emph{\textbf{L}},\emph{\textbf{l}})=\frac{D({\emph{\textbf{x}}},{\emph{\textbf{W}}})}{   \mathop {\max }\limits_{l_i} D({\emph{\textbf{x}}},{\emph{\textbf{W}}})  }= \frac { \sum\nolimits_{i=1}^n  p_i  W_i D_f(L_i,l_i)}  { \sum\nolimits_{i=1}^n  p_i  W_i }=  \frac { \sum\nolimits_{i=1}^n  p_i  W_i \frac{r^{L_i-l_i}-1}{r^{L_i}-1}}  { \sum\nolimits_{i=1}^n  p_i  W_i },   % D_r({\emph{\textbf{P}}},\emph{\textbf{L}},\emph{\textbf{l}},{\emph{\textbf{W}}})=
\end{flalign} 
where $\emph{\textbf{L}}=\{L_1,...,L_n\}$ and $\emph{\textbf{l}}=\{l_1,...,l_n\}$.

\subsection{Problem Formulation}
\subsubsection{General Storage System}
In fact, the available storage space can then be expressed as $\sum\nolimits_{i=1}^n p_i l_i$.
For each given target maximum available storage space constraint $\sum\nolimits_{i=1}^n p_i l_i \le T$, we shall optimize storage resources allocation strategy of this system by minimizing the RWRE, which can be expressed as
\begin{flalign}\label{equ:optimal storage strategy}
  \mathcal{P}_1: \,\, \mathop {\min }\limits_{l_i}\,\,\, &   D_r({\emph{\textbf{x}}},{\emph{\textbf{W}}}) \\
%\mathcal{P}_1: \,\, \mathop {\min }\limits_{l_i}\,\,\, &\frac{1}{ \sum\nolimits_{j=1}^n  p_j  W_j  } \sum\limits_{i=1}^n  p_i  W_i \frac{r^{L_i-l_i}-1}{r^{L_i}-1}  \\
\textrm{s.t.}\,\,\,& \sum\limits_{i=1}^n p_i l_i \le T \tag{\theequation a}\label{equ:optimal storage strategy a}\\
%& \sum\nolimits_i^{n} {p_i = 1}  \tag{\theequation b}\label{equ:optimal storage strategy b}\\
& 0\le l_i \le L_i \,\,\textrm{for} \,\,i=1,2,...,n .\tag{\theequation b}\label{equ:optimal storage strategy b}
\end{flalign}
The storage systems which can be characterized by Problem $  \mathcal{P}_1$ are referred to as \textit{general storage system}.

\begin{rem}
In fact, this paper focuses on allocating resources by category with taking message importance into account, while the conventional source coding searches the shortest average description length of a random variable. 
\end{rem}

\subsubsection{Ideal Storage system}
In practice, the storage size of raw data is the same frequently for ease of use.
Thus, we focus on the case where the original storage size is the same for simplifying the analysis in this paper, and use $L$ to denote it (i.e., $L_i=L$ for $i=1,2,...,n$). As a result, we have
\begin{flalign}\label{equ: initial_Dr_L}
%& \mathop {\min }\limits_{l_i}\,\,\, \frac{1}{ \sum\nolimits_{j=1}^n  p_j  W_j  } \sum\limits_{i=1}^n  p_i  W_i \frac{r^{L_i-l_i}-1}{r^{L_i}-1}  \nonumber \\
\mathop {\min }\limits_{l_i}   D_r({\emph{\textbf{x}}},{\emph{\textbf{W}}}) =  \frac{r^{L}  \mathop {\min }\limits_{l_i} \sum\limits_{i=1}^n  p_i  W_i {r^{-l_i}}   }{(r^{L}-1) \sum\limits_{i=1}^n p_i W_i}  - \frac{1}{r^{L}-1} .
 \end{flalign}
 Thus, the problem $\mathcal{P}_1$ can be rewritten as
\begin{flalign}\label{equ:optimal storage strategy1}
\mathcal{P}_2: \,\, \mathop {\min }\limits_{l_i}\,\,\, & \sum\limits_{i=1}^n  p_i  W_i  r^{-l_i}  \\
\textrm{s.t.}\,\,\,& \sum\limits_{i=1}^n p_i l_i \le T \tag{\theequation a}\label{equ:optimal storage strategy1 a}\\
%& \sum\nolimits_i^{n} {p_i = 1}  \tag{\theequation b}\label{equ:optimal storage strategy1 b}\\
& 0\le l_i \le L \,\,\textrm{for} \,\,i=1,2,...,n .\tag{\theequation b}\label{equ:optimal storage strategy1 b}
\end{flalign}
Since we will mainly focus on the characteristics of the solutions in Problem $\mathcal{P}_2$ in this paper, we use \textit{ideal storage system} to represent this model in later sections of this paper.

\subsubsection{Quantification Storage System}
A \textit{quantification storage system} quantizes and stores the real data acquired from sensors in the real world.
The data is usually a real number, which requires infinite number bits to describe it accurately.
That is, the original storage size of each class approaches the infinite number, (i.e., $L_i =L\to +\infty$ for $i=1,2,...,n$), in this case.
As a result, the RWRE can be rewritten as
\begin{flalign}\label{equ: RWRE in NMIM}
D_r({\emph{\textbf{x}}},{\emph{\textbf{W}}})=\mathop {\lim }\limits_{L \to \infty} \left\{ \frac{ \sum\limits_{i=1}^n  p_i  W_i {r^{-l_i}}   }{(1-r^{-L}) \sum\limits_{i=1}^n p_i W_i}  - \frac{1}{r^{L}-1} \right\}=\frac{  \sum\limits_{i=1}^n  p_i  W_i {r^{-l_i}}   }{ \sum\limits_{i=1}^n p_i W_i}.
\end{flalign}
Therefore, the problem $\mathcal{P}_1$ in this case is reduced to
\begin{flalign}\label{equ:optimal storage strategy2}
\mathcal{P}_3: \,\, \mathop {\min }\limits_{l_i}\,\,\, & \sum\limits_{i=1}^n  p_i  W_i  r^{-l_i}  \\
\textrm{s.t.}\,\,\,& \sum\limits_{i=1}^n p_i l_i \le T \tag{\theequation a}\label{equ:optimal storage strategy2 a}\\
%& \sum\nolimits_i^{n} {p_i = 1}  \tag{\theequation b}\label{equ:optimal storage strategy1 b}\\
&  l_i \ge 0 \,\,\textrm{for} \,\,i=1,2,...,n .\tag{\theequation b}\label{equ:optimal storage strategy2 b}
\end{flalign}

\section{Optimal Allocation Strategy with Limited Storage Space}\label{sec: optimal problem}
In this section, we shall first solve the problem $\mathcal{P}_1$ and give the solutions. In fact, the solutions provide the optimal storage space allocation strategy for digital data on the best effort in minimizing the RWRE when the total available storage size is limited. Then, the problem $\mathcal{P}_2$ will be solved, whose solution characterizes the optimal storage space allocation strategy with the same original storage size. Moreover, we shall also discuss the solution in the case where the original storage size of each class approaches the infinite number by studying the problem $\mathcal{P}_3$.

\subsection{Optimal Allocation Strategy in General Storage System}
\begin{thm}\label{thm:optimal strategy general}
For a storage system with probability distribution $(p_1,p_2,...,p_n)$, $L_i$ is the storage size of the raw data of the class $i$ for $i=1,2,...,n$.
For a given maximum available storage space $T$ ($0 \le T \le \sum\nolimits_{i=1}^n p_i L_i$), when the radix is $r$ ($r>1$), the solution of Problem $\mathcal{P}_1$ is given by
%\begin{equation}\label{equ:thm optimal strategy}
%l^*_i = \min \left( {\left\lfloor \frac{T-\sum\nolimits_{j=1}^{\tilde N_L} p_{T_j} L_{T_j}}{ \sum\nolimits_{j=1}^{\tilde N} p_{I_j}} +\frac{ \ln W_i  }{\ln r}-\frac{  {\sum\nolimits_{j=1}^{\tilde N}  p_{I_j} \ln W_{I_j}}  }{\ln r { \sum\nolimits_{j=1}^{\tilde N} p_{I_j}} }  \right\rfloor}^+, L_i \right),
%\end{equation}
\begin{flalign} \label{equ:thm optimal strategy general}
l_i= &\left\{
   \begin{aligned}
 & \quad\quad\quad\quad 0\quad\quad\quad\quad\quad\quad\quad\quad\quad\quad\quad\quad\quad\,\,\,\,\,\,\, \textit{if} \,\, l_i<0,\\
 &\frac{ \ln(\ln r) +\ln W_i -\ln (1-r^{-L_i}) - \ln \lambda^* }{\ln r} \quad\quad \textit{if} \,\, 0\le l_i \le L_i,\\
 & \quad\quad\quad\quad L_i \quad\quad\quad\quad\quad\quad\quad\quad\quad\quad\quad\quad\quad\,\,\,\,\,\textit{if} \,\, l_i>L_i,
   \end{aligned}
   \right.
\end{flalign}
where $\lambda^*$ is chosen so that $\sum\nolimits_{i=1}^n p_i l_i=T$.
\end{thm}
\begin{proof}
By means of Lagrange multipliers and Karush-Kuhn-Tucher conditions, when ignoring the constant $\sum\nolimits_{i=1}^n p_i W_i$, we set up the functional
\begin{flalign}\label{equ: proff kkt}
{\nabla _l}\left\{ { \sum\limits_{i=1}^n  p_i  W_i \frac{r^{L_i-l_i}-1}{r^{L_i}-1}} +\lambda^* (\sum\limits_{i=1}^n p_i l_i - T)+ \mu_1 (l_1-L_1)+...+\mu_n (l_n-L_n) \right\}=0& \\
\sum\limits_{i=1}^n p_i l_i - T=0&  \tag{\theequation a} \label{equ: proff kkt a}\\
\mu_i(l_i-L_i)=0 \,\,\textrm{for} \,\,i=1,2,...,n  &\tag{\theequation b} \label{equ: proff kkt b}\\
l_i-L_i \le 0 \,\,\textrm{for} \,\,i=1,2,...,n   &\tag{\theequation c} \label{equ: proff kkt c} \\
\mu_i \ge 0 \,\,\textrm{for} \,\,i=1,2,...,n   &\tag{\theequation d} \label{equ: proff kkt d} \\
l_i \ge 0 \,\,\textrm{for} \,\,i=1,2,...,n   &\tag{\theequation e} \label{equ: proff kkt e} 
\end{flalign}
Hence, we obtain
\begin{flalign}\label{equ: li_langrange}
l_i=\frac{\ln p_i +\ln(\ln r) +\ln W_i -\ln (1-r^{-L_i}) - \ln (\lambda^* p_i+\mu_i)}{\ln r}.
\end{flalign}

First, it is easy to check that Equation (\ref{equ: proff kkt b})-(\ref{equ: proff kkt d}) hold when $\mu_i=0$ and $l_i \le L_i$. Hence, we have
\begin{flalign}\label{equ: li_langrange1}
l_i=\frac{ \ln(\ln r) +\ln W_i -\ln (1-r^{-L_i}) - \ln \lambda^* }{\ln r}.
\end{flalign}

Second, if $l_i$ in Equation (\ref{equ: li_langrange}) is larger than $L_i$, we will have $\mu_i>0$ and $l_i = L_i$ due to Equation (\ref{equ: proff kkt b})-(\ref{equ: proff kkt d}).

Third, if $l_i<0$, we will let $l_i=0$ according to Equation (\ref{equ: proff kkt e}).

Moreover, $\lambda^*$ is chosen so that $\sum\nolimits_{i=1}^n p_i l_i=T$ due to Equation (\ref{equ: proff kkt a}).

Therefore, based on the discussion above, we get Equation (\ref{equ:thm optimal strategy general}) in order to ensure $0\le l_i \le L_i$.
%\begin{flalign}
%l_i= &\left\{
   %\begin{aligned}
 %& \quad\quad 0,\quad\quad\quad\quad\quad\quad\quad\quad\quad\,\,\, l_i<0,\\
 %&\frac{\ln(\ln r) +\ln W_i - \ln (\lambda)}{\ln r}, 0\le l_i \le L_i,\\
 %& \quad\quad L_i, \quad\quad\quad\quad\quad\quad\quad\quad\quad\,l_i>L_i.
   %\end{aligned}
   %\right.   \nonumber
%\end{flalign}
\end{proof}

\begin{rem}\label{rem: def of Ij Tj}
Let ${\tilde N}$ be the number of $l_i$ which meets $0\le l_i \le L_i$ and $\{I_j, j=1,2,...,\tilde N\}$ is part of the sequence of $\{1,2,...,N\}$ which satisfies $0 \le  { \ln(\ln r) +\ln W_{I_j} -\ln (1-r^{-L_{I_j} }) - \ln \lambda^* }\le L_{I_j}{\ln r} $. Furthermore, $\{T_j, j=1,2,...,\tilde N_L\}$ is used to denote the part of the sequence of $\{1,2,...,N\}$ which satisfies $ { \ln(\ln r) +\ln W_{T_j} -\ln (1-r^{-L_{T_j}}) - \ln \lambda^* }>L_{T_j}{\ln r} $. 
\end{rem}

%$l_i=\min\left({\left( \frac{\ln(\ln r) +\ln W_i - \ln \lambda }{\ln r} \right)}^+,L_i\right)$, where $(x)^+$ is equal to $x$ when $x>0$, and it is zero when $x \le 0$.
Substituting Equation (\ref{equ:thm optimal strategy general}) in the constraint $\sum\nolimits_{i=1}^n p_i l_i = T$, we have
\begin{flalign} \label{equ:lambda general}
\ln \lambda^*=\ln\ln r+\frac{ \sum\nolimits_{j=1}^{\tilde N} p_{I_j}\ln W_{I_j} -\sum\nolimits_{j=1}^{\tilde N} p_{I_j} \ln (1-r^{-L_{I_j}})   -\ln r(T-\sum\nolimits_{j=1}^{\tilde N_L} p_{T_j} L_{T_j})  } {\sum\nolimits_{j=1}^{\tilde N} p_{I_j} }.
\end{flalign}

Hence, for $0 \le l_i \le L$, we obtain
\begin{flalign} \label{equ: li for W}
l_i=\frac{T-\sum\nolimits_{j=1}^{\tilde N_L} p_{T_j} L_{T_j}}{ \sum\nolimits_{j=1}^{\tilde N} p_{I_j}}+\frac{\ln W_i}{\ln r} -\frac{  {\sum\nolimits_{j=1}^{\tilde N}  p_{I_j} \ln W_{I_j}}   }{\ln r { \sum\nolimits_{j=1}^{\tilde N} p_{I_j}}}   +\frac{  {\sum\nolimits_{j=1}^{\tilde N}  p_{I_j} \ln (1-r^{-L_{I_j}}) }  }{\ln r { \sum\nolimits_{j=1}^{\tilde N} p_{I_j}}}.
\end{flalign}

In fact, $T$, $p_i$, $r$, $L_i$ are usually constraints for a given recommendation system, and therefore $l_i$ is only determined by the second and the third items on the right side of Equation (\ref{equ: li for W}), which means the storage size depends on the message importance and the probability distribution of class for given available storage size.
\begin{rem}\label{rem: l_star_floor general}
Since the actual compressed storage size $l^*_i$ must be integer, the actual storage size allocation strategy is
\begin{flalign}\label{equ: l_star_floor general}
l^*_i = \min \left( {\left\lfloor \frac{T-\sum\nolimits_{j=1}^{\tilde N_L} p_{T_j} L_{T_j}    }{ \sum\nolimits_{j=1}^{\tilde N} p_{I_j}}+\frac{\ln W_i}{\ln r} -\frac{  {\sum\nolimits_{j=1}^{\tilde N}  p_{I_j} \ln W_{I_j}}   }{\ln r { \sum\nolimits_{j=1}^{\tilde N} p_{I_j}}}   +\frac{  {\sum\nolimits_{j=1}^{\tilde N}  p_{I_j} \ln (1-r^{-L_{I_j}}) }  }{\ln r { \sum\nolimits_{j=1}^{\tilde N} p_{I_j}}}  \right\rfloor}^+, L \right),
\end{flalign}
where $(x)^+$ is equal to $x$ when $x \ge 0$, and it is zero when $x <0$. In addition, $\lfloor x \rfloor$ is the largest integer smaller than or equal to $x$.
\end{rem}

%\begin{rem}\label{rem: li_NL0}
%When ${\tilde N}_L=0$, there is no $l_i$ in (\ref{equ: li_langrange}) larger than the initial size of storage space $L_i$, and $\mu_i$ in (\ref{equ: li_langrange}) is always zero. Therefore, we have
%\begin{flalign}
%l^*_i =  {\left\lfloor \frac{T}{ \sum\nolimits_{j=1}^{\tilde N} p_{I_j}} +\frac{ \ln W_i  }{\ln r}-\frac{  {\sum\nolimits_{j=1}^{\tilde N}  p_{I_j} \ln W_{I_j}}  }{\ln r { \sum\nolimits_{j=1}^{\tilde N} p_{I_j}} }  \right\rfloor}^+.
%\end{flalign}
%\end{rem}

\subsection{Optimal Allocation Strategy in Ideal Storage System}
Then, we pay attention to the case where the original storage size is the same for simplifying the analysis. Based on Theorem \ref{thm:optimal strategy general}, we get the following corollary in ideal storage system.
\begin{cor}\label{cor:thm optimal strategy}
For a storage system with probability distribution $(p_1,p_2,...,p_n)$, the original storage size of each class is the same, which is given by $L_i=L$ for $i=1,2,...,n$.
For a given maximum available storage space $T$ ($0 \le T \le  L$), when the radix is $r$ ($r>1$), the solution of Problem $\mathcal{P}_2$ is given by
\begin{flalign} \label{equ:thm optimal strategy}
l_i= &\left\{
   \begin{aligned}
 & \quad\quad 0\quad\quad\quad\quad\quad\quad\quad\,\,\,\, \quad\,\,  \textrm{if} \,\, l_i<0,\\
 &\frac{ \ln(\ln r) +\ln W_i  - \ln \lambda }{\ln r}\quad\,\,\,\, \textrm{if} \,\, 0\le l_i \le L,\\
 & \quad\quad L \quad\quad\quad\quad\quad\quad\quad\,\,\,\,\,\quad \textrm{if} \,\, l_i>L,
   \end{aligned}
   \right.
\end{flalign}
where $\lambda$ is chosen so that $\sum\nolimits_{i=1}^n p_i l_i=T$.
\end{cor}
\begin{proof}
Let $\lambda=\lambda^*(1-r^{-L})$ and $L_i=L$ for $i=1,2,...,n$. Substituting them in Equation (\ref{equ:thm optimal strategy general}), we find that $l_i$ in this case can be rewritten as Equation (\ref{equ:thm optimal strategy}).
\end{proof}

Substituting Equation (\ref{equ:thm optimal strategy}) in the constraint $\sum\nolimits_{i=1}^n p_i l_i = T$, we obtain
\begin{flalign} \label{equ:lambda}
\ln \lambda=\ln\ln r+\frac{ \sum\nolimits_{j=1}^{\tilde N} p_{I_j}\ln W_{I_j}-\ln r(T-T_{N_L})  } {\sum\nolimits_{j=1}^{\tilde N} p_{I_j} },
\end{flalign}
where $\tilde N$, $\tilde N_L$, $I_j$, $T_j$ is still given by Remark \ref{rem: def of Ij Tj} with letting $\lambda^*=\lambda^*(1-r^{-L})$. In addition, $T_{N_L}=\sum\nolimits_{j=1}^{\tilde N_L} p_{T_j} L$.
Hence, for $0 \le l_i \le L$, we obtain
\begin{flalign} 
l_i=\frac{T-T_{N_L }}{ \sum\nolimits_{j=1}^{\tilde N} p_{I_j}}+\frac{\ln W_i}{\ln r} -\frac{  {\sum\nolimits_{j=1}^{\tilde N}  p_{I_j} \ln W_{I_j}}   }{\ln r { \sum\nolimits_{j=1}^{\tilde N} p_{I_j}}}.
\end{flalign}
\begin{rem}\label{rem: l_star_floor}
Since the actual compressed storage size $l^*_i$ must be integer, the actual storage size allocation strategy is
\begin{flalign}\label{equ: l_star_floor}
l^*_i = \min \left( {\left\lfloor \frac{T-T_{N_L}}{ \sum\nolimits_{j=1}^{\tilde N} p_{I_j}} +\frac{ \ln W_i  }{\ln r}-\frac{  {\sum\nolimits_{j=1}^{\tilde N}  p_{I_j} \ln W_{I_j}}  }{\ln r { \sum\nolimits_{j=1}^{\tilde N} p_{I_j}} }  \right\rfloor}^+, L \right).
\end{flalign}
\end{rem}

\begin{rem}\label{rem:li_Nn}
When ${\tilde N}=n$, $0 \le l_i \le L$ always holds for $1\le i\le n$, and the actual storage size is given by
\begin{flalign}\label{equ: li_Nn}
l^*_i=\left\lfloor T+\frac{\ln W_i - {\sum\nolimits_{i=1}^{n}  p_i \ln W_i  } }{\ln r}   \right\rfloor.
\end{flalign}
\end{rem}

In order to illustrate the geometric interpretation of this algorithm, we might as well take
\begin{flalign}
\beta=\frac{\ln \ln r -\ln \lambda}{\ln r},
\end{flalign}
and the optimal storage size can be simplified to
\begin{flalign} \label{equ:thm optimal strategy water-filling}
l_i= &\left\{
   \begin{aligned}
 & \quad\quad 0,\quad\quad\quad\quad\quad\quad\quad\,\,\,\, \textrm{if} \,\, \beta -\frac{\ln (1/W_i)}{\ln r}<0.\\
 & \beta -\frac{\ln (1/W_i)}{\ln r}\quad\quad\,\,\,\quad\quad \textrm{if} \,\, 0\le \beta -\frac{\ln (1/W_i)}{\ln r} \le L.\\
 & \quad\quad L, \quad\quad\quad\quad\quad\quad\quad\,\,\,\,\textrm{if} \,\, \beta -\frac{\ln (1/W_i)}{\ln r}>L.
   \end{aligned}
   \right.
\end{flalign}

The monotonicity of optimal storage size with respect to importance weight is discussed in the following theorem.
\begin{thm}\label{thm:p_l_W_monotonicity}
Let $(p_1,p_2,...,p_n)$ be a probability distribution and ${{\textrm{W}}}={W_1,...,W_n}$ be importance weights. $L$ and $r$ are fixed positive integers ($r>1$). The solution of Problem $\mathcal{P}_2$ meets: $l_i  \ge l_j$ if $W_i> W_j$ for $\forall i,j \in \{1,2,...,n\}$.
\end{thm}
\begin{proof}
Refer to the Appendix \ref{Appendices p_l_W_monotonicity}.
\end{proof} 

This gives rise to a kind of restrictive water-filling, which is presented in Figure \ref{fig:paper_AL_MIM_watering}. Choose a constant $\beta$ so that $\sum\nolimits_{i=1}^n p_i l_i=T$. The storage size depends on the difference between $\beta$ and $\frac{\ln (1/W_i)}{\ln r}$. In Figure \ref{fig:paper_AL_MIM_watering}, we obtain that $\beta$ characterizes the height of water surface, and $\frac{\ln (1/W_i)}{\ln r}$ determines the bottom of the pool. Actually, no storage space is assigned to the data with this difference less than zero. When the difference is in the interval $[0,L]$, this difference is exactly the storage size. Furthermore, the storage size will be truncated to $L$ bits if the difference is larger than $L$. Compared with the conventional water-filling, the lowest height of the bottom of the pool is constricted in this restrictive water-filling.

\begin{rem}
The restrictive water-filling in Figure \ref{fig:paper_AL_MIM_watering} is summarized as follows. 
\begin{itemize}[leftmargin=*,labelsep=5.8mm]
\item   For the data with extremely small message importance, $\frac{\ln (1/W_i)}{\ln r}$ is so large that the bottom of the pool is above the water surface. Thus, the storage size of this kind of data is zero.  
\item   For the data with small message importance, $\frac{\ln (1/W_i)}{\ln r}$ is large, and therefore the bottom of the pool is high. Thus, the storage size of this kind of data is small.  
\item   For the data with large message importance, $\frac{\ln (1/W_i)}{\ln r}$ is small, and therefore the bottom of the pool is low. Thus, the storage size of this kind of data is large.
\item   For the data with extremely large message importance, $\frac{\ln (1/W_i)}{\ln r}$ is so small that the bottom of the pool is constricted in order to truncate the storage size to $L$.
\end{itemize}
\end{rem}

Thus, this optimal storage space allocation strategy is a high efficiency adaptive storage allocation algorithm for the fact that it can make rational use of all the storage space according to message importance to minimize the RWRE.
\begin{figure}[H]
  \centerline{\includegraphics[width=10.0cm]{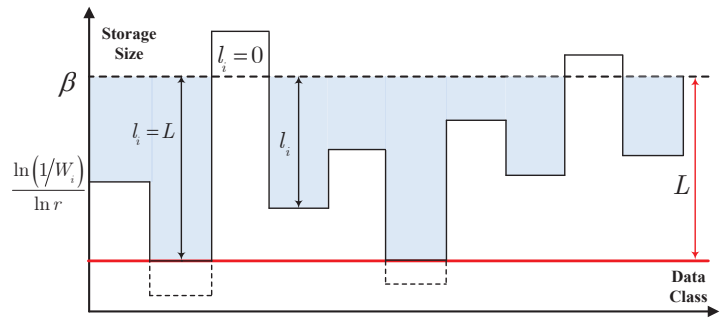}}
  %\centerline{(c) Result 3}
  \caption{Restrictive water-filling for optimal storage size.}\label{fig:paper_AL_MIM_watering}
\end{figure}

This solution can be gotten by means of recursive algorithm in practice, which is shown in Algorithm \ref{alg:recursive}, where we define an auxiliary function as
\begin{flalign} \label{equ: the optimal same storage size}
f(i,\emph{\textbf{W}},\emph{\textbf{p}},L,T,r,K_{\min},K_{\max})=\frac{T }{ \sum\nolimits_{j=K_{\min}}^{K_{\max} } p_j }+\frac{\ln W_i}{\ln r} -\frac{  {\sum\nolimits_{j=K_{\min}}^{K_{\max}}  p_{j} \ln W_{j}}   }{\ln r { \sum\nolimits_{j=K_{\min}}^{K_{\max}} p_{j}}}.
\end{flalign}

\begin{algorithm}
\caption{Storage Space Allocation Algorithm}
\label{alg:recursive}
\begin{algorithmic}[1]
\REQUIRE ~~\\
The message importance, $\emph{\textbf{W}}=\{W_i,i=1,2,...,n\}$ (Sort it to satisfy $W_1 \ge W_2 \ge ... \ge W_n$)\\
\vspace{-2mm}
The probability distribution of source, $\emph{\textbf{P}}=\{p_i, i=1, 2, ..., n\}$\\
\vspace{-2mm}
The original storage size, $L$\\
\vspace{-2mm}
The maximum available storage space, $T$\\
\vspace{-2mm}
The radix, $r$\\
\vspace{-2mm}
The auxiliary variables, $K_{\min},K_{\max}$ (Let $K_{\min}=1,K_{\max}=n$ as the original values)\\
\vspace{-2mm}
\ENSURE ~~\\
\vspace{-2mm}
The compressed code length, $\emph{\textbf{l}}=\{l_i,i=K_{\min},...,K_{\max}\}$\\
\vspace{-2mm}
Denote the following algorithm as $\{l_{K_{\min}},...,l_{K_{\max}}\}=\phi(\emph{\textbf{W}},\emph{\textbf{P}},L,T,r,K_{\min},K_{\max})$\\
\vspace{-1mm}
\STATE    ${l_i}^\prime  \leftarrow f(i,\emph{\textbf{W}},\emph{\textbf{P}},L,T,r,K_{\min},K_{\max})$ for $i=K_{\min},...,K_{\max}$ \,\quad\quad\quad\quad  $\vartriangleright$ See Equation (\ref{equ: the optimal same storage size}) \\
\vspace{-1mm}
\STATE    \textbf{if} $\forall t \in \{K_{\min},...,K_{\max}\}$ such that $0 \le l_{t}^\prime \le L$ and $\sum\nolimits_{i=K_{\min}}^{K_{\max}} p_i l_i^\prime=T $
\vspace{-1mm}
\STATE    ${l_i}\leftarrow  {l_i}^\prime $ for $i=K_{\min},...,K_{\max}$
\vspace{-1mm}
\STATE    \textbf{else} \quad \textbf{if}  $K_{\max}>K_{\min}$
\vspace{-1mm}
\STATE    \quad\quad\quad ${l_i}^{(1)}\leftarrow L$ for $i=1,..,K_{\min}-1$
\vspace{-1mm}
\STATE    \quad\quad\quad${l_i}^{(1)}\leftarrow0$ for $i=K_{\max},...,n$
\vspace{-1mm}
%\vspace{2mm}
%\STATE    \quad\quad\quad $T^\prime \leftarrow T$\\
\STATE    \quad\quad\quad ${l_i}^{(1)}  \leftarrow \phi(\emph{\textbf{W}},\emph{\textbf{P}},L,T,r,K_{\min},K_{\max}-1)$ for $i=K_{\min},...,K_{\max}-1$ 
\vspace{-1mm}   
 \STATE \quad\quad\quad   $\epsilon^{(1)}=D_r(\emph{\textbf{W}},L,\emph{\textbf{l}}^{(1)})$  \quad\quad\quad\quad\quad\quad\quad\quad\quad\quad\quad\quad\quad\quad\quad\quad\quad  $\vartriangleright$ See Equation (\ref{equ:Def_RWRE}) 
 \vspace{-1mm}
\STATE   \quad\quad\quad ${l_i}^{(2)}\leftarrow L$ for $i=1,..,K_{\min}$ 
\vspace{-1mm}
\STATE \quad\quad\quad ${l_i}^{(2)}\leftarrow0$ for $i=K_{\max}+1,...,n$ 
\vspace{-1mm}
 \STATE \quad\quad\quad $T^\prime \leftarrow T- p_{K_{\min}} L$
 \vspace{-1mm}
 \STATE \quad \quad\quad ${l_i}^{(2)}  \leftarrow \phi(\emph{\textbf{W}},\emph{\textbf{P}},L,T^\prime,r,K_{\min}+1,K_{\max})$ for $i=K_{\min}+1,...,K_{\max}$
      \vspace{-1mm}
       \STATE \quad\quad\quad   $\epsilon^{(2)}=D_r(\emph{\textbf{W}},L,\emph{\textbf{l}}^{(2)})$  \quad\quad\quad\quad\quad\quad\quad\quad\quad\quad\quad\quad\quad\quad\quad\quad\quad  $\vartriangleright$ See Equation (\ref{equ:Def_RWRE})
       \vspace{-1mm}
\STATE    \quad\quad\quad\quad\quad \textbf{if}  $\epsilon^{(1)}\ge \epsilon^{(2)}$
\vspace{-1mm}
\STATE    \quad\quad\quad\quad\quad  \quad \quad ${l_i}\leftarrow  {l_i}^{(1)} $ for $i=1,2,...,n$ 
\vspace{-1mm}
\STATE    \quad\quad\quad\quad\quad \textbf{else}  
\vspace{-1mm}
\STATE    \quad\quad\quad\quad\quad\quad\quad ${l_i}\leftarrow  {l_i}^{(2)} $ for $i=1,2,...,n$ \\
\vspace{-1mm}
 \STATE   \quad\quad\quad\quad\quad \textbf{end}
 \vspace{-1mm}
\STATE   \quad\quad\quad \textbf{else}  \\
\vspace{-1mm}
\STATE \quad\quad\quad${l_i}\leftarrow L$ for $i=1,..,K_{\min}-1$
\vspace{-1mm}
\STATE \quad\quad\quad${l_i}\leftarrow0$ for $i=K_{\min}+1,...,n$
\vspace{-1mm}
\STATE   \quad\quad\quad    ${l_{K_{\min}}}\leftarrow (T-\sum\nolimits_{i=1}^{K_{\min}-1} p_i L) / p_{K_{\min}}$ \\
\vspace{-1mm}
\STATE \quad\quad\quad  \textbf{end}  \\
\vspace{-1mm}
 \STATE  \textbf{end}
 \vspace{-1mm}
\RETURN $l_i$ for $i=K_{\min},...,K_{\max}$
\end{algorithmic}
\end{algorithm}

\subsection{Optimal Allocation Strategy in Quantification Storage System}
\begin{cor}
For a given maximum available storage space $T$ ($T \ge 0$), when probability distribution is $(p_1,p_2,...,p_n)$ and the radix is $r$ ($r>1$), the solution of Problem $\mathcal{P}_3$ is given by
\begin{flalign} \label{equ:thm optimal strategy real number}
l_i= \left(  \frac{ \ln(\ln r) +\ln W_i  - \ln \lambda }{\ln r}  \right)^+ ,
\end{flalign}
where $\lambda$ is chosen so that $\sum\nolimits_{i=1}^n p_i l_i=T$.
\end{cor}
\begin{proof}
Let $L \to \infty$ in Corollary \ref{cor:thm optimal strategy}, the solutions in Equation (\ref{equ:thm optimal strategy}) can be simplified to Equation (\ref{equ:thm optimal strategy real number}).
\end{proof}

In fact, the optimal storage space allocation strategy in this case can be seen as a kind of water-filling, which gets rid of the constraint on the lowest height of the bottom of the pool.

\section{Property of Optimal Storage Strategy Based on Message Importance Measure}\label{sec: MIM}
Considering that the ideal storage system can capture most of characteristics of the lossy compression storage model in this paper, we focus on the property of optimal storage strategy in it in this section for ease of analyzing. Specifically, we ignore rounding and adopt $l_i$ in Equation (\ref{equ:thm optimal strategy}) as the optimal storage size of the $i$-th class in this section.
Moreover, we focus on a special kind of the importance weight.
Namely, MIM is adopted as the importance weight in this paper, for the fact that it can effectively measure the cost of the error in data reconstruction in the small-probability event scenarios \cite{liu2017non,she2019importance}. 

\subsection{Normalized MIM}
In order to facilitate comparison under different parameters, the normalized MIM is used and we can write
\begin{flalign}\label{equ: def of W}
W_i=\frac{e^{\varpi (1-p_i)}}{\sum\nolimits_{j=1}^n e^{\varpi(1-p_j)}},
\end{flalign}
where $\varpi$ is the importance coefficient.

Actually, it is easy to check that $0\le W_i \le 1$ for $i=1,2,...,n$. Moreover, it is obvious that the sum of those in all event classes is one.

\subsubsection{Positive Importance Coefficient}
For positive importance coefficient (i.e., $\varpi>0$), let $\alpha_1=\arg \mathop {\min}\limits_i p_i$ and assume $p_{\alpha_1}< p_i$ for $i \ne \alpha_1$. The derivation of it with respect to the importance coefficient is
\begin{flalign}
\frac{\partial W_{\alpha_1}}{\partial \varpi} = \frac{\sum\nolimits_{j=1}^n (p_j-p_{\alpha_1})e^{\varpi(2-p_{\alpha_1}-p_j)} }{ {\left(\sum\nolimits_{j=1}^n e^{\varpi(1-p_j)}\right)}^2 } \ge 0.
\end{flalign}
Therefore, $W_{\alpha_1}$ increases as $\varpi$ increases. In particular, as $\varpi$ approaches positive infinity, we have
\begin{flalign}\label{equ:PQ_min1}
 \mathop {\lim }\limits_{\varpi \to  + \infty} W_{\alpha_1} =&\mathop {\lim }\limits_{\varpi  \to  + \infty } \frac{e^{\varpi (1-p_{\alpha_1})}}{\sum\nolimits_{j=1}^n e^{\varpi(1-p_j)}}\\
   =&\mathop {\lim }\limits_{\varpi  \to  + \infty } \frac{e^{\varpi (1-p_{\alpha_1})}}{e^{\varpi (1-p_{\alpha_1})}+\sum\nolimits_{j \ne \alpha_1} e^{\varpi(1-p_j)}}  \tag{\theequation a}\\
  =&   \mathop {\lim }\limits_{\varpi  \to  + \infty } \frac{1}{1+\sum\nolimits_{j \ne \alpha_1} e^{\varpi(p_{\alpha_1}-p_j)}}  \tag{\theequation b}\\
  =&   1 .\tag{\theequation c} \label{equ:PQ_min1 c}
\end{flalign}
Obviously, $\mathop {\lim }\limits_{\varpi \to  + \infty} W_{i}=0$ for $i \ne \alpha_1$.
\begin{rem}\label{rem: positive Importance coefficient}
As $\varpi$ approaches positive infinity, the importance weight with the smallest probability is one and others are all zero, which means only a fraction of data almost owns almost all of the critical information that users care about in the viewpoint of this message importance.
\end{rem}
\subsubsection{Negative Importance Coefficient}
When importance coefficient is negative (i.e., $\varpi<0$), let $\alpha_2=\arg \mathop {\max}\limits_i p_i$ and assume $p_{\alpha_2}> p_i$ for $i \ne \alpha_2$. The derivation of it with respect to the importance coefficient is
\begin{flalign}
\frac{\partial W_{\alpha_2}}{\partial \varpi} = \frac{\sum\nolimits_{j=1}^n (p_j-p_{\alpha_2})e^{\varpi(2-p_{\alpha_2}-p_j)} }{ {\left(\sum\nolimits_{j=1}^n e^{\varpi(1-p_j)}\right)}^2 } \le 0.
\end{flalign}
Therefore, $W_{\alpha_1}$ decreases as $\varpi$ increases. In particular, as $\varpi$ approaches negative infinity, we have
\begin{flalign}\label{equ:PQ_min1}
 \mathop {\lim }\limits_{\varpi \to  - \infty} W_{\alpha_2} =&\mathop {\lim }\limits_{\varpi  \to  - \infty } \frac{e^{\varpi (1-p_{\alpha_2})}}{\sum\nolimits_{j=1}^n e^{\varpi(1-p_j)}}\\
   =&\mathop {\lim }\limits_{\varpi  \to  - \infty } \frac{e^{\varpi (1-p_{\alpha_2})}}{e^{\varpi (1-p_{\alpha_2})}+\sum\nolimits_{j \ne \alpha_2} e^{\varpi(1-p_j)}}  \tag{\theequation a}\\
  =&   \mathop {\lim }\limits_{\varpi  \to  - \infty } \frac{1}{1+\sum\nolimits_{j \ne \alpha_2} e^{\varpi(p_{\alpha_2}-p_j)}}  \tag{\theequation b}\\
  =&   1 .\tag{\theequation c} \label{equ:PQ_min1 c}
\end{flalign}
Obviously, $\mathop {\lim }\limits_{\varpi \to  - \infty} W_{i}=0$ for $i \ne \alpha_2$.
\begin{rem}\label{rem: negative Importance coefficient}
As $\varpi$ approaches negative infinity, the importance weight with the biggest probability is one and others are all zero. If the biggest probability is not too big, the majority of message importance can also be included in not too much data.
\end{rem}

\subsection{Optimal Storage Size for Each Class}
Assume ${\tilde N}=n$ and ignore rounding, due to Equation (\ref{equ: li_Nn}), we obtain
\begin{flalign}\label{equ: mim length}
l_i=&T+\frac{ \ln \frac{e^{\varpi (1-p_i)}}{\sum\nolimits_{j=1}^n e^{\varpi(1-p_j)}} - \sum\limits_{i=1}^n {p_i \ln \frac{e^{\varpi (1-p_i)}}{\sum\nolimits_{j=1}^n e^{\varpi(1-p_j)}}} }{\ln r}  \\
%l_i=&T+\frac{\varpi (1-p_i) -  \ln {\sum\nolimits_{j=1}^n e^{\varpi(1-p_j)}}  - \sum\limits_{i=1}^n {p_i \varpi (1-p_i)}+ \ln {\sum\nolimits_{j=1}^n e^{\varpi(1-p_j)}} \sum\nolimits_{i=1}^n p_i }{\ln r}  \\
%=&T+\frac{\varpi (1-p_i) - \sum\limits_{i=1}^n {p_i \varpi (1-p_i)} }{\ln r}  \tag{\theequation a} \label{equ: mim length a}\\
=&T+\frac{ \varpi }{\ln r}( { \gamma_p}- p_i),  \tag{\theequation a}\label{equ: mim length a}
\end{flalign}
where $\gamma_p$ is an auxiliary variable and it is given by
\begin{flalign}
\gamma_p=\sum\limits_{i=1}^n p_i^2.
\end{flalign}
In fact, its natural logarithm is the minus R{\'{e}}nyi entropy of order two, i.e., $\gamma_p = e^{-H_2(\emph{\textbf{P}})}$ where $H_2(\emph{\textbf{P}})$ is the R{\'{e}}nyi entropy $H_{\alpha}(\cdot)$ when $\alpha=2$~\cite{van2014renyi}. Furthermore, we have the following lemma on $\gamma_p$.
\begin{lem}\label{lem:p2_lower}
Let $(p_1,p_2,...,p_n)$ be a probability distribution, then we have
\begin{flalign}
\frac{1}{n} &\le \gamma_p \le 1, \\ %\,\,\ \textrm{and}\,\,\,
-\frac{1}{4} &\le \gamma_p -p_i \le 1.   \tag{\theequation a} 
\end{flalign}
\end{lem}
\begin{proof}
Refer to the Appendix \ref{Appendices A}.
\end{proof} 
%\begin{proof}
%It is noted that
%\begin{flalign}
%\sum\limits_{i=1}^n { p_i^2}=\frac{1}{n} \left({\sum\limits_{i=1}^n { p_i^2} \sum\limits_{i=1}^n { 1^2}}\right) \ge \frac{1}{n}  { {\left(\sum\limits_{i=1}^n { p_i}\right)}^2  }=\frac{1}{n},
%\end{flalign}
%where the equality holds only if $(p_1,p_2,...,p_n)$ is uniform distribution. Moreover,
%\begin{flalign}
%\sum\nolimits_{i=1}^n { p_i^2} \le \sum\nolimits_{i=1}^n { p_i}=1,
%\end{flalign}
%where the equality holds only if there is $p_t=1$ ($t\in\{1,2,...,n\}$).
%\end{proof}
Thus, we find $l_i>T$ if $(1/n-p_i)\varpi>0$. Besides, we obtain $l_i=T$ when $p_i= \gamma_p$.

\begin{thm}\label{thm:p_l_monotonicity}
Let $(p_1,p_2,...,p_n)$ be a probability distribution and $W_i={e^{\varpi (1-p_i)}}/{\sum\nolimits_{j=1}^n e^{\varpi(1-p_j)}}$ be importance weight. The optimal storage size in ideal storage system has the following properties:
\begin{enumerate}[leftmargin=*,labelsep=4.9mm]
  \item[(1)] $l_i\ge l_j$ if $p_i< p_j$ for $\forall i,j \in \{1,2,...,n\}$ when $\varpi>0$;
  \item[(2)] $l_i\le l_j$ if $p_i< p_j$ for $\forall i,j \in \{1,2,...,n\}$ when $\varpi<0$.
\end{enumerate}
\end{thm}
\begin{proof}
Refer to the Appendix \ref{Appendices p_l_monotonicity}.
\end{proof} 
\begin{rem}
Due to \cite{she2019importance}, the data with smaller probability usually possesses larger importance when $\varpi>0$, while the data with larger probability usually possesses larger importance when $\varpi<0$.
Therefore, this optimal allocation strategy makes rational use of all the storage space by providing more storage size for paramount data and less storage size for insignificance data. It agrees with the intuitive idea, which is that users generally are more concerned about the data that they need rather than the whole data itself.
\end{rem}

\begin{lem}\label{lem:p1_N}
Let $(p_1,p_2,...,p_n)$ be a probability distribution and $r$ be radix. $L$ and $T$ are integers, and $T<L$. If $\varpi$ meets $0 \le T+{ \varpi(\gamma_p- p_i) }/{\ln r} \le L$, then we have ${\tilde N}=n$.
\end{lem}
\begin{proof}
According to Equation (\ref{equ: mim length a}) and constraint $0 \le T+{ \varpi(\gamma_p- p_i) }/{\ln r} \le L$, we obtain $0\le l_i \le L$ for $\forall i \in \{1,2,...,n\}$. In this case, ${\tilde N}=n$.
\end{proof}

In fact, when $\varpi \ge 0$, due to Equation (\ref{equ: mim length a}) and Lemma \ref{lem:p2_lower}, we obtain
\begin{flalign}\label{equ: w L T 1}
0\le T-\frac{ \varpi }{4\ln r}  \le T+\frac{ \varpi(\gamma_p- p_i) }{\ln r}&\le T+\frac{ \varpi }{\ln r} \le L.
\end{flalign}
Similarly, when $\varpi < 0$, we have
\begin{flalign}\label{equ: w L T 2}
0\le T+\frac{ \varpi }{\ln r}  \le T+\frac{ \varpi(\gamma_p- p_i) }{\ln r}&\le T-\frac{ \varpi }{4\ln r} \le L.
\end{flalign}
According to Equation (\ref{equ: w L T 1}) and Equation (\ref{equ: w L T 2}), we find ${\tilde N}=n$ always holds If $\max(4\ln r (T-L), -T/\ln r) \le \varpi \le \min(4T \ln r, \ln r(L-T))$.

\subsection{Relative Weighted Reconstruction Error} 
For convenience, $D({\emph{\textbf{x}}},\varpi)$ is used to denote $D({\emph{\textbf{x}}},{\emph{\textbf{W}}})$. Due to Equation (\ref{equ: initial_Dr_L}), we have
\begin{flalign}\label{equ: Dr_all}
D_r({\emph{\textbf{x}}},\varpi)=\frac{1}{r^L-1} \left( {  \frac{\sum\nolimits_{i=1}^n  p_i e^{\varpi(1-p_i)} r^{L-l_i} }  {\sum\nolimits_{i=1}^n  p_i e^{\varpi(1-p_i)}  }-1    }\right) .
\end{flalign}
If $T$ is zero, then we will have $l_i=0$ for $i=1,2,...,n$. In this case, $D_r({\emph{\textbf{x}}},\varpi)=1$. On the contrary, $D_r({\emph{\textbf{x}}},\varpi)=0$ when $l_i=L$ for $i=1,2,...,n$.

\begin{thm}\label{thm: monotonicity of Dr}
$D_r({{\textbf{x}}},\varpi)$ has the following properties:
\begin{enumerate}[leftmargin=*,labelsep=4.9mm]
  \item[(1)] $D_r({{\textbf{x}}},\varpi)$ is monotonically decreasing with $\varpi$ in $(0, +\infty)$;
  \item[(2)] $D_r({{\textbf{x}}},\varpi)$ is monotonically increasing with $\varpi$ in $(-\infty, 0)$;
  \item[(3)] $D_r({{\textbf{x}}},\varpi) \le D_r({{\textbf{x}}},0)=(r^{L-T}  -1)/(r^L-1) $.
\end{enumerate}
\end{thm}
\begin{proof}
Refer to the Appendix \ref{Appendices C}.
\end{proof} 

\begin{rem}
As shown in Remark \ref{rem: positive Importance coefficient} and Remark \ref{rem: negative Importance coefficient}, the overwhelming majority of important information will gather in a fraction of data as the importance coefficient increases to negative/positive infinity.
Therefore, we can heavily reduce the storage space with extremely small of RWRE with the increasing of the absolute value of importance coefficient.
In fact, this special characteristic of weight reflects the effect of users' preference.
That is, it is beneficial for data compression that the data that users care about is highly clustered.
Moreover, when $\varpi=0$, all the importance weight is the same, which leads to the incompressibility for the fact that there is no special characteristic of weight for users to make rational use of storage space.
\end{rem}

In the following part of this section, we will discuss the case where $0 \le T+{ \varpi(\gamma_p- p_i) }/{\ln r} \le L$, which means all $l_i$ can be given by Equation (\ref{equ: mim length a}) and $n=\tilde N$ due to Lemma \ref{lem:p1_N}. In this case, substituting Equation (\ref{equ: mim length a}) in Equation (\ref{equ: initial_Dr_L}), the RWRE is
\begin{flalign}\label{equ:weighted reconstruction error}
D_r({\emph{\textbf{x}}},\varpi)=\frac{   e^{\varpi(1-\gamma_p)}r^{\Delta}     }  { (r^L-1)\sum\nolimits_{i=1}^n  {p_i e^{\varpi (1-p_i)}}  }-\frac{1}{r^L-1},
\end{flalign}
where $\Delta=L-T$, which characterizes the actual compressed storage space.

Since that $0 \le T+{ \varpi(\gamma_p- p_i) }/{\ln r} \le L$, we have
\begin{flalign}
\frac{ \varpi(\gamma_p-p_{\alpha_1}) }{ \ln r}\le L-T\le L-\frac{\varpi(p_{\alpha_2}-\gamma_p)}{\ln r} 
\end{flalign}
Hence, 
\begin{flalign}\label{equ: delta12}
\delta_1=\frac{   e^{\varpi(1-p_{\alpha_1})}     }  { (r^L-1)\sum\nolimits_{i=1}^n  {p_i e^{\varpi (1-p_i)}}  }-\frac{1}{r^L-1} \le D_r({\emph{\textbf{x}}},\varpi) \le \frac{   e^{\varpi(1-p_{\alpha_2})}r^{L}     }  { (r^L-1)\sum\nolimits_{i=1}^n  {p_i e^{\varpi (1-p_i)}}  }-\frac{1}{r^L-1}=\delta_2.
\end{flalign}

\begin{thm}\label{thm: max delta}
For a given storage system with the probability distribution of data sequence ${\textbf{P}}=(p_1,p_2,...,p_n)$, let $L$, $r$ be fixed positive integers ($r>1$), and $\varpi$ meets $0 \le T+{ \varpi(\gamma_p- p_i) }/{\ln r} \le L$ for $i=1,2,...,n$.
For giving upper bound of the RWRE $\delta$ ($\delta_1 \le \delta \le \delta_2$ where $\delta_1$ and $\delta_1$ is defined in Equation (\ref{equ: delta12})), the maximum available compressed storage size $\Delta^*(\delta)$ is given by
\begin{flalign}\label{equ:max delta}
\Delta^*(\delta)&=\frac{\ln \left(1+\delta (r^L-1)\right)+L(\varpi,{\textbf{P}}) -\varpi+\varpi \gamma_p  }{\ln r}  \\
         &\ge \frac{\ln \left(1+\delta (r^L-1)\right)  }{\ln r},  \tag{\theequation a} \label{equ:max delta a}
\end{flalign}
where $L(\varpi,{\textbf{P}})=\ln \sum\nolimits_{i=1}^n {p_i e^{\varpi (1-p_i)}}$, and the equality of (\ref{equ:max delta a}) holds if the probability distribution of data sequence is uniform distribution or the importance coefficient is zero.
\end{thm}
\begin{proof}
It is easy to check that ${\tilde N}=n$ according to Lemma \ref{lem:p1_N} for the fact that $0 \le T+{ \varpi(\gamma_p- p_i) }/{\ln r} \le L$.
Let $D({\emph{\textbf{x}}},\varpi) \le \delta$. By means of Equation (\ref{equ:weighted reconstruction error}), we solve this inequality and obtain
\begin{flalign} 
\Delta \le \frac{\ln \left(1+\delta (r^L-1)\right)+L(\varpi,\emph{\textbf{p}}) -\varpi+\varpi \gamma_p  }{\ln r}=\Delta^*(\delta) ,
\end{flalign}
%$\Delta^*(\delta)$ is the maximum compressed storage space which makes the weighted reconstruction error no more than $\delta$. It is easy to check that $\Delta^*(\delta)$ increases with increasing with $\delta$. Furthermore, it has the following proposition.
where $L(\varpi,\emph{\textbf{p}})=\ln \sum\nolimits_{i=1}^n {p_i e^{\varpi (1-p_i)}}$. Then we have the following inequality:
\begin{flalign}
\Delta^*(\delta)  \overset{(a)}  \ge\frac{\ln \left(1+\delta (r^L-1)\right)+\ln  {e^{\sum\nolimits_{i=1}^n p_i \varpi (1-p_i)}} -\varpi+\varpi \gamma_p  }{\ln r}   \nonumber  
%&= \frac{\ln \left(1+\delta (r^L-1)\right)+  {{\sum\limits_{i=1}^n p_i \varpi (1-p_i)}} -\varpi+\varpi \gamma_p  }{\ln r}   \nonumber \\
 =\frac{\ln \left(1+\delta (r^L-1)\right)  }{\ln r},  
\end{flalign}
where $(a)$ follows from Jensen's inequality. Since the exponential function is strictly convex, the equality holds only if $\varpi(1-p_i)$ is constant everywhere, which means $(p_1,p_2,...,p_n)$ is uniform distribution or importance coefficient $\varpi$ is zero.
\end{proof} 

\begin{rem}
In conventional source coding, the encoding length depends on the entropy of sequence, and a sequence is incompressible if its probability distribution is uniform distribution \cite{Elements}. In Theorem \ref{thm: max delta}, the uniform distribution is also worst case, since the system achieves the minimum compressed storage size. Although the focus is different, they both show that the uniform distribution is detrimental for compression.
\end{rem}

Furthermore, it is also noted that
\begin{flalign}\label{equ:max delta 2}
\Delta^*(\delta)\le \Delta^*(\delta_2)=L+\frac{\varpi(\gamma_p-p_{\alpha_2})}{\ln r} \le L,
\end{flalign}
for the fact that $\gamma_p \le p_{\alpha_2} $. In order to make $\Delta^*(\delta_2)$ approaches $L$, $\gamma_p-p_{\alpha_2}$ should be as close to zero as possible in the range which $0 \le T+{ \varpi(\gamma_p- p_i) }/{\ln r} \le L$ for $i=1,2,...,n$ holds. 

When the importance coefficient is constant, for two probability distributions $\emph{\textbf{P}}$ and $\emph{\textbf{Q}}$, if $L(\varpi,\emph{\textbf{P}}) +\varpi \gamma_p>L(\varpi,\emph{\textbf{Q}}) +\varpi \gamma_q$, then we will obtain $\Delta^*$ in $\emph{\textbf{P}}$ is larger than that in $\emph{\textbf{Q}}$.
In fact, $L(\varpi,\emph{\textbf{p}})$ is defined as MIM in~\cite{fan2016message}, and $\gamma_p = e^{-H_2(\emph{\textbf{P}})}$~\cite{van2014renyi}. Thus, the maximum available compressed storage size is under the control of MIM and R{\'{e}}nyi entropy of order two.
For typical small-probability event scenarios where there is a exceedingly small probability, the MIM is usually large, and $\gamma_p$ is also not small simultaneously with big probability. Therefore, $\Delta^*(\delta)$ is usually large in this case. As a result, much more compressed storage space can be gotten in typical small-probability event scenarios while compared to that in uniform probability distribution. Namely, the data can compressed by means of the characteristic of the typical small-probability events, which may help to improve the design of practical storage systems in big data.

\section{Property of Optimal Storage Strategy Based on Non-parametric Message Importance Measure}\label{sec: NMIM}
In this section, we define the importance weight based on the form of non-parametric message importance measure (NMIM) to characterize the RWRE~\cite{liu2017non}. Then, the importance weight in this section is given by
\begin{flalign}\label{equ: W in NMIM}
W_i=\frac{e^{{(1-p_i)}/{p_i}}}{  \sum\nolimits_{j=1}^n e^{{(1-p_j)}/{p_j}}  }.
\end{flalign}

Due to Equation (\ref{equ: l_star_floor}), the optimal storage size in ideal storage system by this importance weight is given by
\begin{flalign}
l^*_i &= \min \left( {\left\lfloor \frac{T-T_{N_L}}{ \sum\limits_{j=1}^{\tilde N} p_{I_j}} +\frac{ 1 }{p_i\ln r} -\frac{1}{\ln r}-\frac{\ln \sum\limits_{j=1}^n e^{{(1-p_j)}/{p_j}}}{  \ln r}   -\frac{  {\sum\limits_{j=1}^{\tilde N}  (1-p_{I_j} -p_{I_j}\ln \sum\limits_{j=1}^n e^{{(1-p_j)}/{p_j}} )    }  }{\ln r { \sum\limits_{j=1}^{\tilde N} p_{I_j}} }  \right\rfloor}^+, L\right)  \nonumber\\
&=\min \left( {\left\lfloor \frac{T-T_{N_L}}{ \sum\limits_{j=1}^{\tilde N} p_{I_j}} +\frac{ 1  }{p_i\ln r} -\frac{  {\tilde N       }  }{\ln r { \sum\limits_{j=1}^{\tilde N} p_{I_j}} }  \right\rfloor}^+, L\right) .
\end{flalign}

For two probabilities $p_i$ and $p_j$, if $p_i<p_j$, then we will have $W_i>W_j$. Thus, we obtain $l_i^* \ge l_j^*$ according to Theorem \ref{thm:p_l_W_monotonicity}.

Assume ${\tilde N}=n$ and ignore rounding, due to Equation (\ref{equ: li_Nn}), we obtain
\begin{flalign}
l^*_i= T+\frac{1}{p_i \ln r}-\frac{n}{\ln r}.
\end{flalign}
Let $0\le l_i \le L$, we find
\begin{flalign} \label{equ:NMIM l-p}
\frac{1}{n+(L-T)\ln r}   \le p_i \le &\left\{
   \begin{aligned}
 &\frac{1}{n-T\ln r} \quad \textrm{if} \,\,\,   n>T \ln r.\\
 &1 \,\,\,\,\,\,\quad\quad\quad\quad \textrm{if} \,\,\, n\le T \ln r.\\
   \end{aligned}
   \right. 
\end{flalign}
Generally, this constraint does not invariably hold, and therefore we usually do not have ${\tilde N}=n$. 

Substituting Equation (\ref{equ: W in NMIM}) in Equation (\ref{equ: RWRE in NMIM}), the RWRE is given by
\begin{flalign}
D_r({\emph{\textbf{x}}},{\emph{\textbf{W}}})=\frac{  \sum\nolimits_{i=1}^n  p_i e^{(1-p_i)/p_i} {r^{-l_i}}   }{ \sum\nolimits_{i=1}^n p_i e^{(1-p_i)/p_i} }.
\end{flalign}

For the quantification storage system as shown in $\mathcal{P}_3$ in this section, if the maximum available storage size satisfies $n \le T \ln r $, arbitrary probability distribution will make Equation (\ref{equ:NMIM l-p}) hold, which means ${\tilde N}=n$. In this case, the RWRE can be expressed as
\begin{flalign}\label{equ: Dr NMIM}
D_r({\emph{\textbf{x}}},{\emph{\textbf{W}}})=  e^{n-1-\mathcal{L}({\textbf{\emph{P}}})} r^{-T}     ,
\end{flalign}
where $\mathcal{L}({\textbf{\emph{P}}})=\ln \sum\nolimits_{i=1}^n p_i e^{(1-p_i)/p_i} $, which is defined as the NMIM~\cite{liu2017non}.

It is noted that $D_r({\emph{\textbf{x}}},{\emph{\textbf{W}}})=0$ as $T$ approaches positive infinity.
Since $n \le T \ln r $, we find $D_r({\emph{\textbf{x}}},{\emph{\textbf{W}}})  \le r^{-1-\mathcal{L}({\textbf{\emph{P}}})}$.
Furthermore, since that $ \mathcal{L}({\textbf{\emph{P}}}) \ge n-1$ according to Ref. \cite{liu2017non}, we obtain $D_r({\emph{\textbf{x}}},{\emph{\textbf{W}}})  \le r^{-n}$. Let $D_r({\emph{\textbf{x}}},{\emph{\textbf{W}}}) \le \delta$, we have
\begin{flalign}
T \ge \frac{n-1-\mathcal{L}({\textbf{\emph{P}}}) -\ln \delta  }{\ln r}.
\end{flalign}

Furthermore, due to Ref. \cite{liu2017non}, $\mathcal{L}({\textbf{\emph{P}}}) \approx \ln p_{\alpha_1} e^{\frac{1-p_{\alpha_1}}{p_{\alpha_1}}}$ when $p_{\alpha_1}$ is small. Hence, for small $p_{\alpha_1}$, the RWRE in this case can be reduced to
\begin{flalign}
D_r({\emph{\textbf{x}}},{\emph{\textbf{W}}})\approx \frac{ e^{n-1/p_{\alpha_1}}} {p_{\alpha_1} }     r^{-T}   . 
\end{flalign}
It is easy to check that $D_r({\emph{\textbf{x}}},{\emph{\textbf{W}}})$ increases as $p_{\alpha_1}$ increases in this case.

Obviously, for a giving RWRE, the minimum required storage size for the quantification storage system decreases with increasing of $\mathcal{L}({\textbf{\emph{P}}})$. That is to say, the data with large NMIM will get large compression ratio. 
In fact, the NMIM in the typical small-probability event scenarios is generally large according to Ref. \cite{liu2017non}.
Thus, this compression strategy is effective in the typical small-probability event scenarios.

\section{Numerical Results}\label{Numerical Results}
We now present numerical results to validate the results in this paper. For ease of illustrating, we ignore rounding and adopt $l_i$ in (\ref{equ:thm optimal strategy}) as the optimal storage size of the $i$-th class.
\subsection{Optimal Storage Size Based on MIM in Ideal Storage System}
The broken line graph of the optimal storage size is shown in Figure \ref{fig:fig1}, when the probability distribution is $P=(0.03, 0.07, 0.1395, 0.2205, 0.25, 0.29)$. In fact, $0.2205 \approx \gamma_P$ and $1/n\approx 0.167$. The available storage size $T$ is $4$ bits, and the original storage size of each data is $10$ bits. The importance coefficients are given by $\varpi_1=-35,\varpi_2=-10,\varpi_3=0,\varpi_4=10,\varpi_5=35$ respectively.
Some observations can be obtained. 
When $\varpi>0$, the optimal storage size of the $i$-th class decreases with the increasing of its probability. 
On the contrary, the optimal storage size of the $i$-th class increases as its probability increases when $\varpi<0$. 
Besides, the optimal storage size is invariably $T$ ($T=4$) when $\varpi=0$.
Furthermore, $l_i$ increases as $\varpi$ increases for $i=1,2,3$, and it decreases with $\varpi$ for $i=5,6$.
For small importance coefficient ($\varpi_2,\varpi_3,\varpi_4$), $0<l_i<L$ holds for $i=1,2,...,6$, and $l_4$ is extremely close to $T$ ($T=4$).
\begin{figure}[htb!]
\centering
\includegraphics[width=10 cm]{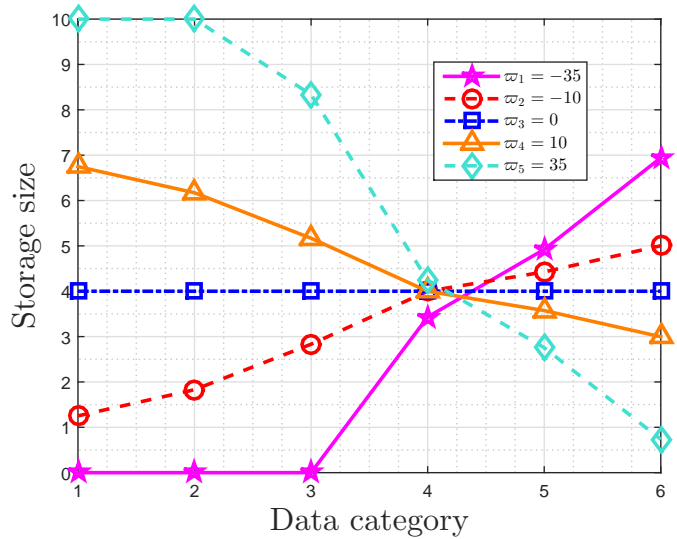}
  %\centerline{(c) Result 3}
  \caption{Broken line graph of optimal storage size with the probability distribution $(0.03, 0.07, 0.1395, 0.2205, 0.25, 0.29)$, for giving maximum available storage size $T=4$ and original storage size $L=10$.}\label{fig:fig1}
\end{figure}

\subsection{The Property of the RWRE Based on MIM in Ideal Storage System}
Then we focus on the properties of the RWRE. The available storage size $T$ is varying from $0$ to $8$ bits, and the original storage size of each data is $16$ bits. 
Figure \ref{fig:fig2} and Figure \ref{fig:fig3} both present the relationship between the RWRE and the available storage size $\Delta$ with the probability distribution $(0.031,0.052,0.127,0.208,0.582)$.

\begin{figure}[htb!]
\centering
\includegraphics[width=10 cm]{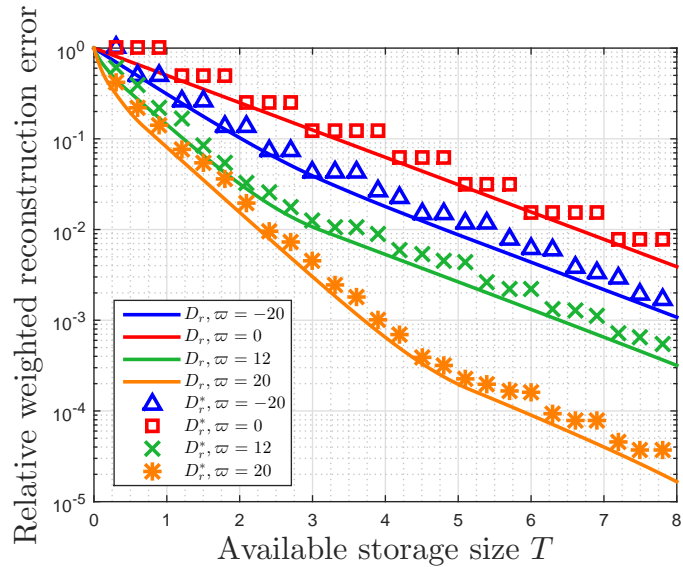}
  %\centerline{(c) Result 3}
  \caption{RWRE versus available storage size $T$ with the probability distribution $(0.031,0.052,0.127,0.208,0.582)$ in the case of the value of importance coefficient $\varpi=-20,0,-12,20$. $D_r$ is acquired by substituting Equation (\ref{equ:thm optimal strategy}) in Equation (\ref{equ: Dr_all}), while $D_r^*$ is obtained by substituting Equation (\ref{equ: l_star_floor}) in Equation (\ref{equ: Dr_all}).}\label{fig:fig2}
\end{figure}
\begin{figure}[htb!]
\centering
\includegraphics[width=10 cm]{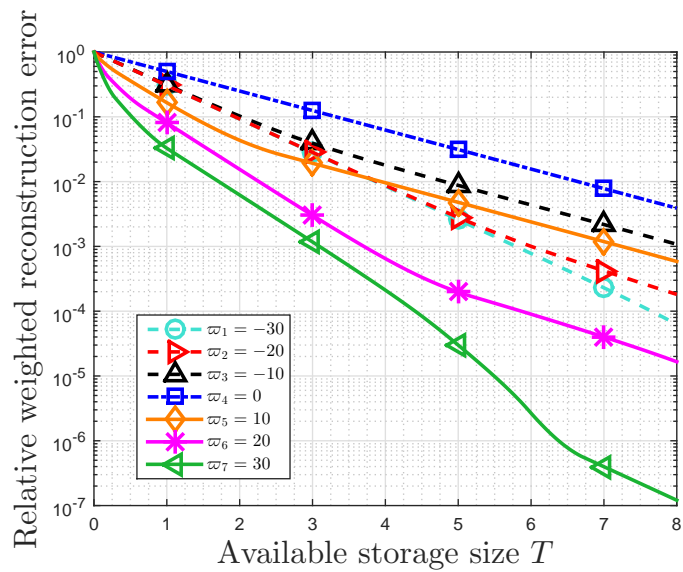}
  %\centerline{(c) Result 3}
  \caption{RWRE $D_r({\emph{\textbf{x}}},\varpi)$ versus available storage size $T$ with the probability distribution $(0.031,0.052,0.127,0.208,0.582)$ in the case of the value of importance coefficient $\varpi=-30,-20,-10,0,10,20,30$.}\label{fig:fig3}
\end{figure}
Figure \ref{fig:fig2} focuses on the error of RWRE by rounding number with different importance coefficient $\varpi$ ($\varpi=-20,0,-12,20$). In Figure \ref{fig:fig2}, the RWRE $D_r$ is acquired by substituting Equation (\ref{equ:thm optimal strategy}) in Equation (\ref{equ: Dr_all}), while the RWRE $D_r^*$ is obtained by substituting Equation (\ref{equ: l_star_floor}) in Equation (\ref{equ: Dr_all}). In this figure, $D_r^*$ has tierd descent as the available storage size increases, while $D_r$ monotonically decreases with increasing of the available storage size.
Figure \ref{fig:fig2} shows that $D_r$ is always less than or equal to $D_r^*$ and they are very close to each other for the same importance coefficient, which means that $D_r$ can be used as the lower bound of $D_r^*$ to reflect the characteristics of $D_r^*$.

Furthermore, some other observations can be obtained in Figure \ref{fig:fig3}. For the same $T$, the RWRE increases as $\varpi$ increases when $\varpi<0$, while the RWRE decreases with increasing of $\varpi$ when $\varpi>0$. Besides, the RWRE is the largest when $\varpi=0$.
It is also observed that the RWRE always decreases with increasing of $T$ for giving $\varpi$. Besides, for any importance coefficient, the RWRE will be $1$ if available storage size is zero. Generally, there is a trade-off between the RWRE and the available storage size, and the results in this paper propose an alternative lossy compression strategy based on message importance.

Then let the importance coefficient $\varpi$ be $5$ and the available storage size $T$ be varying from $2$ to $8$ bits. In addition, the original storage size is still $16$ bits. Besides, the compressed storage size is given by $\Delta=L-T$. 
In this case, Figure \ref{fig:fig4} shows that the RWRE versus compressed storage size $\Delta$ for different probability distributions. The probability distributions and some auxiliary variables are listed in Table \ref{tab:result1}.
Obviously, all probability distributions satisfy $0 \le T+{ \varpi(\gamma_p- p_i) }/{\ln r} \le L$.
It is observed that the RWRE always increases with increasing of $\Delta$ for a giving probability distribution.
Some other observations are also obtained. For the same $\Delta$, the RWRE of uniform distribution is the largest all the time. 
Furthermore, if the RWRE is required to be less than a specified value, which is exceedingly common in actual system in order to make the difference between the raw data and the stored data accepted, the maximum available compressed storage space increases with increasing of $L(\varpi,\emph{\textbf{P}}) +\varpi e^{-H_2(\emph{\textbf{P}})}$.
Besides, the maximum available compressed storage space is the smallest in uniform distribution.
As an example, when the RWRE is required to be smaller than $0.01$, the maximum available compressed storage space of $P_1$, $P_2$, $P_3$, $P_4$, $P_5$ is $11.85$, $10.97$, $9.99$, $9.73$, $9.36$ respectively.
In particular, the maximum available compressed storage size in uniform distribution is the smallest, which suggests the data with uniform distribution is incompressible.
\begin{figure}[htb!]
\centering
\includegraphics[width=10 cm]{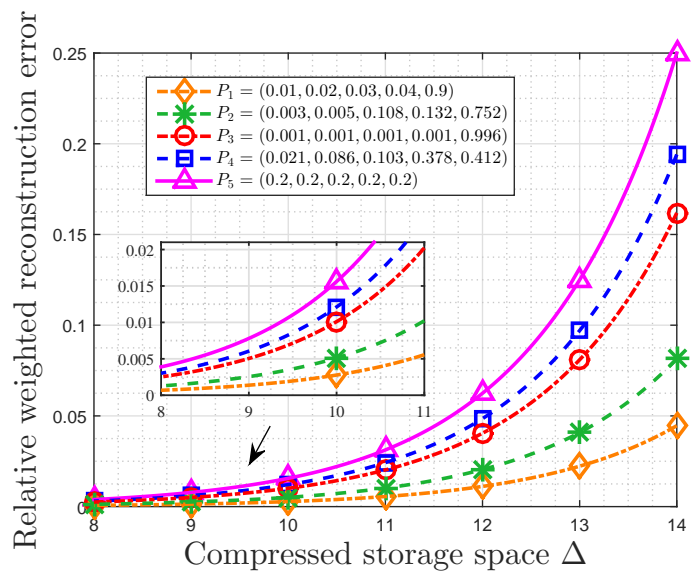}
  %\centerline{(c) Result 3}
   \caption{RWRE $D_r({\emph{\textbf{x}}},\varpi)$ vs. compressed storage size $\Delta$ with importance coefficient $\varpi=5$.}\label{fig:fig4}
\end{figure}
\begin{table}[htb!]
\centering
    \caption{The auxiliary variables in ideal storage system.}\label{tab:result1}
\begin{tabular}{ccccc}
\toprule
\textbf{Variable }&   \textbf{Probability distribution} &\boldmath{${ \varpi(\gamma_p-p_{\alpha_1})}/{\ln r}$} &\boldmath{${ \varpi(\gamma_p-p_{\alpha_2})}/{\ln r}$} &  \boldmath{$L(\varpi,\emph{\textbf{P}}) +\varpi e^{-H_2(\emph{\textbf{P}})}$}  \\
\midrule
{$P_1$} &  $(0.01,0.02,0.03,0.04,0.9)$  & 5.7924 &  -0.6276  & 6.7234\\
{$P_2$} & $(0.003,0.007,0.108,0.132,0.752)$   & 4.2679 & -1.1350  &6.1305\\
{$P_3$} &  $(0.001,0.001,0.001,0.001,0.996)$  & 7.1487 &  -0.0287&  5.4344\\
{$P_4$} &  $(0.021,0.086,0.103,0.378,0.412)$  &2.2367 & -0.5838  &5.2530\\
{$P_5$} &  $(0.2,0.2,0.2,0.2,0.2)$  & 0 & 0 & 5\\
\bottomrule
\end{tabular}
\end{table}

\subsection{The Property of the RWRE Based on NMIM in Quantification Storage System}
Afterwards, Figure \ref{fig:fig5} presents the relationship between the RWRE and available storage size $T$ for different probability distributions in the quantification storage system. The probability distributions and some auxiliary variables are listed in Table \ref{tab:result2}.
Some observations can be obtained. 
First, the RWRE always decreases with increasing of the available storage size for a giving probability distribution, and there is a trade-off between the RWRE and the available storage size. 
When the available storage size is small ($T<n /\ln r$), the RWRE decreases largely compared to the case where $T$ is large.
Besides, when the maximum available storage size is large ($T>n /\ln r$), the difference between these RWRE remains the same at logarithmic Y-axis. In fact, according to Equation (\ref{equ: Dr NMIM}), this difference between two probabilities in this figure is the difference of NMIM divided by $\log10$. As an example, the difference between ${\textbf{\emph{P}}}_1$ and ${\textbf{\emph{P}}}_4$ in this figure is 30, which satisfies this conclusion for the fact that $(\mathcal{L}({\textbf{\emph{P}}}_1)-\mathcal{L}({\textbf{\emph{P}}}_4))  /\log 10 \approx 30$. Moreover, the RWRE in ${\textbf{\emph{P}}}_1$ is very close to that in ${\textbf{\emph{P}}}_2$, and the minimum probabilities in these two probability distributions are the same, i.e., $p_{\alpha_1}=0.007$. It suggests that the data with the same minimum probability will have the same compression performance no matter how the distribution changes, if the minimum probability  is small.
In addition, it is also observed that the RWRE decreases as NMIM $\mathcal{L}({\textbf{\emph{P}}})$ increases for the same $T$, which means this compression strategy is effective in the large NMIM cases.
\begin{figure}[H]
\centering
\includegraphics[width=10 cm]{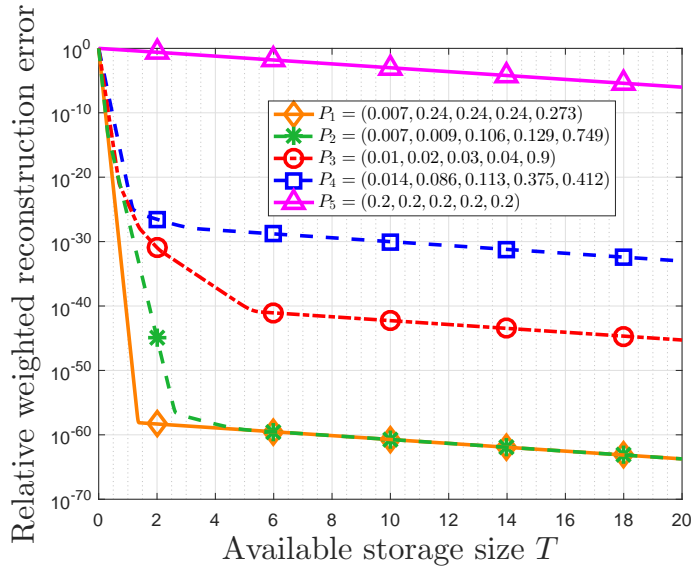}
  %\centerline{(c) Result 3}
   \caption{RWRE versus the available storage size $T$.}\label{fig:fig5}
\end{figure}
\begin{table}[H]
\centering
    \caption{The auxiliary variables in quantification storage system.}\label{tab:result2}
\begin{tabular}{cccc}
\toprule
\textbf{Variable }&   \textbf{Probability distribution} &\boldmath{$p_{\alpha_1}$}  &  \boldmath{$\mathcal{L}({\textbf{\emph{P}}})$}  \\
\midrule
{$P_1$} &  $(0.007,0.24,0.24,0.24,0.273)$  & 0.007 &  136.8953  \\
{$P_2$} & $(0.007,0.009,0.106,0.129,0.749)$   & 0.007 &  136.8953  \\
{$P_3$} &  $(0.01,0.02,0.03,0.04,0.9)$  & 0.01 &  94.3948  \\
{$P_4$} &  $(0.014,0.086,0.113,0.375,0.412)$  &0.014 & 66.1599  \\
{$P_5$} &  $(0.2,0.2,0.2,0.2,0.2)$  & 0.2 & 4.0000 \\
\bottomrule
\end{tabular}
\end{table}

\section{Conclusion}\label{sec: conclusion}
In this paper, we focused on the problem of lossy compression storage from the perspective of message importance when the reconstructed data pursues the least error with certain restricted storage size.
We started with importance-weighted reconstruction error to model the compression storage system, and formulated this problem as an optimization problem
for digital data based on it.
 We gave the solutions by a kind of restrictive water-filling, which presented a alternative way to design an effective storage space adaptive allocation strategy. 
In fact, this optimal allocation strategy prefers to provide more storage size for crucial event classes in order to make rational use of resources, which agrees with the individuals' cognitive mechanism.

Then, we presented the properties of this strategy based on MIM detailedly. 
It is obtained that there is a trade-off between the RWRE and available storage size.
Moreover, the compression performance of this storage system improves as the absolute value of importance coefficient increases.
This is due to the fact that a fraction of data can contain the overwhelming majority of useful information that exerts a tremendous fascination on users as the importance coefficient approaches negative/positive infinity, which suggests that the users' interest is highly-concentrated.
On the other hand, the probability distribution of event classes also has effect on the compression results. When the useful information is only highly enriched in only a small portion of raw data naturally from the viewpoint of users, such as the small-probability event scenarios, it is obvious that we can compress the data greatly with the aid of this characteristics of distribution.
Besides, the properties of storage size and RWRE based on non-parametric MIM were also discussed.
In fact, the RWRE in the data with uniform information distribution was invariably the largest in any case.
Therefore, this paper harbors the idea that the data with uniform information distribution is incompressible, which satisfies the results in information theory.

Proposing more general distortion measure between the raw data and the compressed data, which is no longer only apply to digital data, and using it to acquire the high-efficiency lossy data compression systems from the perspective of message importance are of our future interests.

%In future works, we will investigate more useful properties of NMIM and find the better design of big data storage or transmission strategies.
%%附录
\appendix
\section{}
\subsection{Proof of Theorem \ref{thm:p_l_W_monotonicity}}\label{Appendices p_l_W_monotonicity}
In fact, Equation (\ref{equ:thm optimal strategy water-filling}) can be rewritten as
\begin{flalign} \label{equ:thm optimal strategy water-filling monotonicity}
l_i= &\left\{
   \begin{aligned}
 & \quad\quad 0\quad\quad\quad\quad\quad\quad\quad\,\,\,\,\,\, \textrm{if} \,\, W_i < e^{-\beta \ln r}.\\
 & \beta -\frac{-\ln W_i}{\ln r}\quad\quad\quad\,\,\quad\quad \textrm{if} \,\, e^{-\beta \ln r} \le W_i \le e^{(L-\beta) \ln r}.\\
 & \quad\quad L \quad\quad\quad\quad\quad\quad\quad\,\,\,\,\,\,\textrm{if} \,\, W_i>e^{(L-\beta) \ln r}.
   \end{aligned}
   \right.
\end{flalign}
When $W_i>W_j$, we have
\begin{flalign} 
l_i-l_j= &\left\{
   \begin{aligned}
   &  0\quad\quad \quad\quad\quad\quad\quad\,\,\,   \textrm{if} \,\,\, p_i>e^{(L-\beta) \ln r} ,\,\,\, p_j>e^{(L-\beta) \ln r}.  \\
   &  L- \beta -\frac{\ln W_j}{\ln r}\quad\,\,\,\,\, \textrm{if} \,\,\, p_i>e^{(L-\beta) \ln r} ,\,\,\, e^{-\beta \ln r}\le p_j \le e^{(L-\beta) \ln r} .\\ 
    &  L\quad\quad\quad\quad\quad\quad\,\,\,\,\,\,\,\, \textrm{if} \,\,\, p_i>e^{(L-\beta) \ln r} ,\,\,\, p_j <e^{-\beta \ln r}.\\
    & \frac{\ln W_i -\ln W_j}{\ln r} \quad \,\,\,\,\,\, \textrm{if} \,\,\, e^{-\beta \ln r}\le p_i \le e^{(L-\beta) \ln r}  ,\,\,\, e^{-\beta \ln r}\le p_j \le e^{(L-\beta) \ln r}.\\
 &  \beta -\frac{-\ln W_i}{\ln r} \quad\quad\,\,\,\,  \textrm{if} \,\,\, e^{-\beta \ln r}\le p_i \le e^{(L-\beta) \ln r}  ,\,\,\, p_j <e^{-\beta \ln r}.\\
 &  0\quad\quad \quad\quad\quad\quad\quad\,\,\, \textrm{if} \,\,\, p_i <e^{-\beta \ln r} ,\,\,\, p_j <e^{-\beta \ln r}.
   \end{aligned}
   \right.
\end{flalign}
Due to Equation (\ref{equ:optimal storage strategy1 b}), we obtain that $0\le  \beta -\frac{-\ln W_i}{\ln r} \le L$, and therefore $L- \beta -\frac{\ln W_j}{\ln r} \ge 0$.
Besides, it is easy to check that $\frac{\ln W_i -\ln W_j}{\ln r}$ since that $W_i>W_j$. Thus, $l_i-l_j \ge 0$ if $W_i> W_j$ for $\forall i,j \in \{1,2,...,n\}$. The proof is completed.

\subsection{Proof of Lemma \ref{lem:p2_lower}}\label{Appendices A}
(1) For $\gamma_p$, it is noted that
\begin{flalign}
\sum\limits_{i=1}^n { p_i^2}=\frac{1}{n} \left({\sum\limits_{i=1}^n { p_i^2} \sum\limits_{i=1}^n { 1^2}}\right) \ge \frac{1}{n}  { {\left(\sum\limits_{i=1}^n { p_i}\right)}^2  }=\frac{1}{n},
\end{flalign}
where the equality holds only if $(p_1,p_2,...,p_n)$ is uniform distribution. Moreover,
\begin{flalign}
\sum\nolimits_{i=1}^n { p_i^2} \le \sum\nolimits_{i=1}^n { p_i}=1,
\end{flalign}
where the equality holds only if there is only $p_t=1$ ($t\in\{1,2,...,n\}$) and $p_k =0$ for $k \ne t$.

(2) For $\gamma_p-p_i$, we have $\sum\nolimits_{i=1}^n p_i^2-p_i \le \sum\nolimits_{i=1}^n p_i^2 \le 1$. We have equality if and only if $p_t=1$ and $p_i=0$ for $i \ne t$. Therefore, we only need to check $\sum\nolimits_{i=1}^n p_i^2-p_i \ge -1/4$.

First, if $n=1$, we obtain $\sum\nolimits_{i=1}^n p_i^2-p_i=0$.

Second, if $n=2$, we obtain $\sum\nolimits_{i=1}^n p_i^2-p_i=2(p_1-3/4)^2-1/8$. It is easy to check that $ \sum\nolimits_{i=1}^n p_i^2-p_i \ge -1/8$.

Third, if $n>3$, we use the method of Lagrange multipliers. Let 
\begin{flalign}\label{equ:optimal p2}
J(p)= \sum\limits_{j=1}^n p_j^2-p_i -\lambda (\sum\limits_{j=1}^n p_j -1) .
\end{flalign}
Setting the derivative to $0$, we obtain
\begin{flalign}
2 p^*_j -\lambda &=0   \,\,\, \textrm{for}\,j \ne i      \\
2 p^*_j -1-\lambda &=0\,\,\, \textrm{for}\,j = i.  \tag{\theequation a}
\end{flalign}
Substituting $p^*_j$ in the constraint $\sum\nolimits_{j=1}^n p^*_j =1$, we have
\begin{flalign}
\frac{\lambda (n-1)}{2}+\frac{\lambda +1}{2}=1.
\end{flalign}
Hence, we find $\lambda=1/n$ and 
\begin{flalign}
p^*_j= &\left\{
   \begin{aligned}
 &\frac{n+1}{2n}   \quad\quad \textrm{if}\,\,j=i,\\
 &\frac{1}{2n}\quad\quad\, \quad \textrm{if}\,\,  j \ne i.
   \end{aligned}
   \right.
\end{flalign}
In this case, we get
\begin{flalign}
\sum\limits_{j=1}^n p_j^2-p_i = { \frac{n-1}{4n^2} }+\frac{  {(n+1)}^2 }{4n^2}-\frac{n+1}{2n} =\frac{-n^2+n}{4n^2} \ge -\frac{1}{4}.
\end{flalign}
Thus, Lemma \ref{lem:p2_lower} is proved.

\subsection{Proof of Theorem \ref{thm:p_l_monotonicity}}\label{Appendices p_l_monotonicity}
(1) First, let $p_i<p_j$ when $\varpi>0$. It is noted that
\begin{flalign}
   W_i=\frac{e^{\varpi (1-p_i)}}{\sum\nolimits_{k=1}^n e^{\varpi(1-p_k)}} >\frac{e^{\varpi (1-p_j)}}{\sum\nolimits_{k=1}^n e^{\varpi(1-p_k)}} =W_j.
\end{flalign}
Therefore, we find $l_i\ge l_j$ since that $W_i>W_j$, due to Theorem \ref{thm:p_l_W_monotonicity}.

(2) Second, let $p_i<p_j$ when $\varpi<0$. It is noted that
\begin{flalign}
   W_i=\frac{e^{\varpi (1-p_i)}}{\sum\nolimits_{k=1}^n e^{\varpi(1-p_k)}} <\frac{e^{\varpi (1-p_j)}}{\sum\nolimits_{k=1}^n e^{\varpi(1-p_k)}} =W_j.
\end{flalign}
Therefore, we find $l_i\le l_j$ since that $W_i<W_j$, due to Theorem \ref{thm:p_l_W_monotonicity}. The proof is completed.

\subsection{Proof of Theorem \ref{thm: monotonicity of Dr}}\label{Appendices C}
We define an auxiliary function as 
%$f(\varpi)=  ({\sum\nolimits_{i=1}  p_i e^{\varpi(1-p_i)} r^{L-l_i} }) / ({\sum\nolimits_{i=1}^n  p_i e^{\varpi(1-p_i)}  })$.
\begin{flalign} 
f(\varpi)=  \frac{\sum\nolimits_{i=1}  p_i e^{\varpi(1-p_i)} r^{-l_i} } {\sum\nolimits_{j=1}^n  p_j e^{\varpi(1-p_j)}  }.
\end{flalign}
According to Equation (\ref{equ: Dr_all}), it is noted that the the monotonicity of $D_r({\emph{\textbf{x}}},\varpi)$ with respect to $\varpi$ is the same with that of $f(\varpi)$.

Without loss of generality, let $l_i$ of $p_i$ be
\begin{flalign} 
l_i= &\left\{
   \begin{aligned}
    & \quad\quad\, L \quad\quad\quad\quad\quad\quad\quad\,\,\,\,\,\,\textrm{if} \,\, i=1,2,...,t_1, \\
 &\frac{ \ln(\ln r) +\ln W_i  - \ln \lambda }{\ln r} \quad \textrm{if} \,\, i=t_1+1,...,t_2,\\
 & \quad\quad\, 0\quad\quad\quad\quad\quad\quad\quad\,\,\,\,\,\,\, \textrm{if} \,\,i=t_2+1,t_2+2,...,n,
   \end{aligned}
   \right.
\end{flalign}
where $\lambda$ is given by Equation (\ref{equ:lambda}) where $\{T_j, j=1,...,\tilde N_L\}=\{1,2,...,t_1\}$ and $\{I_j, j=1,...,\tilde N\}=\{t_1+1,...,t_2\}$.

The derivation of $l_i$ with respect to $\varpi$ is given by
\begin{flalign} 
l'_i= &\left\{
   \begin{aligned}
 &\frac{    \sum\nolimits_{k=t_1+1}^{t_2}  p_k(p_k-p_i)    }{  \ln r(\sum\nolimits_{k=t_1+1}^{t_2} p_k)}  \quad\, \textrm{if}\,\,\, i=t_1+1,...,t_2.\\
 & \quad\quad 0 \quad\quad\quad\quad\quad\quad\,\,\,\quad \textrm{else}.
   \end{aligned}
   \right.
\end{flalign}
Hence,
\begin{flalign} 
f'(\varpi)&=\frac{\sum\nolimits_{i}  {\sum\nolimits_{j}   p_i p_j  {  e^{\varpi(2-p_i-p_j)}  } r^{-l_i}  (p_j-p_i-l'_i \ln r)        }  } {  {\left(\sum\nolimits_{j} p_j e^{\varpi(1-p_j)}\right)}^2 }=\frac{F_1+F_2 } {  {\left(\sum\nolimits_{j} p_j e^{\varpi(1-p_j)}\right)}^2 } ,
\end{flalign}
where $F_1=\sum\nolimits_{j}   p_i p_j  {  e^{\varpi(2-p_i-p_j)}  } r^{-l_i}  (p_j-p_i)$ and $F_2=\sum\nolimits_{j}   p_i p_j  {  e^{\varpi(2-p_i-p_j)}  } r^{-l_i}  (-l'_i \ln r)$.

(1) When $\varpi>0$, we have
\begin{flalign} \label{equ: F1}% \label{equ: F1_zero}
%F_1=&\sum\nolimits_{i}  \sum\nolimits_{j}   p_i p_j  {  e^{\varpi(2-p_i-p_j)}  }r^{-l_i} (p_j-p_i)        \\
F_1&=  \sum\limits_{p_j< p_i}   p_i p_j  {  e^{\varpi(2-p_i-p_j)}  } r^{-l_i} (p_j-p_i)    + \sum\limits_{p_j> p_i}   p_i p_j  {  e^{\varpi(2-p_i-p_j)}  } r^{-l_i}  (p_j-p_i)      \\%\tag{\theequation a} \label{equ: F1_zero a} \\
&   \le     \sum\limits_{p_j< p_i}   p_i p_j  {  e^{\varpi(2-p_i-p_j)}  } r^{-l_j} (p_j-p_i)    + \sum\limits_{p_j > p_i}   p_i p_j  {  e^{\varpi(2-p_i-p_j)}  }  r^{-l_i} (p_j-p_i) \tag{\theequation a}\label{equ: F1 a} \\% \tag{\theequation b} \label{equ: F1_zero b} \\
 &=     \sum\limits_{p_j< p_i}   p_i p_j  {  e^{\varpi(2-p_i-p_j)}  } r^{-l_j} (p_j-p_i)    + \sum\limits_{ p_i>p_j}   p_i p_j  {  e^{\varpi(2-p_i-p_j)}  } r^{-l_j} (p_i-p_j)   \tag{\theequation b}\label{equ: F1 b} \\%  \tag{\theequation c} \label{equ: F1_zero c}  \\
&=     \sum\limits_{p_j< p_i}   M_{i,j} r^{-l_j} (p_j-p_i+p_i-p_j) \tag{\theequation c} \label{equ: F1 c}\\
&   =0.  \tag{\theequation d}  \label{equ: F1 d}%\tag{\theequation d} \label{equ: F1_zero d}
\end{flalign}
In fact, if $p_i>p_j$, then we will have $l_i\le l_j$ due to Theorem \ref{thm:p_l_monotonicity}. Thus, $ r^{-l_i} (p_j-p_i)\le  r^{-l_j} (p_j-p_i)$ in this case. With taking $p_i p_j  {  e^{\varpi(2-p_i-p_j)}  }\ge 0$ into account, we have Equation (\ref{equ: F1 a}).
Equation (\ref{equ: F1 b}) is obtained by exchanging the notation of subscript in the second item.

For $t_1<i\le t_2$ and $1\le j \le n$, we have
\begin{flalign} \label{equ: F2_zero}
F_2&=\sum\limits_{i=t_1+1}^{t_2}  \sum\limits_{j=1}^n   p_i p_j  {  e^{\varpi(2-p_i-p_j)}  } r^{-l_i} (-l'_i \ln r)         \\
&=\sum\limits_{i=t_1+1}^{t_2}  \sum\limits_{j=1}^n   p_i p_j {  e^{\varpi(1-p_j)-\ln \ln r +\ln \lambda} } (-l'_i \ln r)     \tag{\theequation a} \\
&= \sum\limits_{j=1}^n  \left(p_j  B_j \left(\sum\limits_{i=t_1+1}^{t_2} p_i  (-l'_i \ln r) \right)\right)  \tag{\theequation b}\\%\tag{\theequation a} \label{equ: F2_zero a}\\
%&= \sum\limits_{j}  \left(p_j  B_j \left(\sum\limits_{i=t_1+1}^{t_2} p_i  \frac{    -\sum\limits_{k=t_1+1}^{t_2}  p_k(p_k-p_i)    }{  \sum\limits_{k=t_1+1}^{t_2} p_k}  \right)\right) \nonumber \\% \tag{\theequation b} \label{equ: F2_zero b}\\
&= \sum\limits_{j=1}^n  \left(p_j  B_j \left(  \frac{    \sum\limits_{i=t_1+1}^{t_2} p_i^2 \sum\limits_{k=t_1+1}^{t_2} p_k-\sum\limits_{k=t_1+1}^{t_2}  p_k^2\sum\limits_{i=t_1+1}^{t_2} p_i    }{  \sum\limits_{k=t_1+1}^{t_2} p_k}  \right)\right)   \tag{\theequation c} \\ %\tag{\theequation c} \label{equ: F2_zero c} \\
&=0, \tag{\theequation d}%\tag{\theequation d} \label{equ: F2_zero d}
\end{flalign}
where $ B_j={  \textrm{exp}\{ {\varpi(1-p_j)-\ln \ln r +\ln \lambda} } \}$.

Based on the discussions above, we have
  \vspace{-1mm}
\begin{flalign} 
f'(\varpi)&=\frac{F_1+F_2 } {  {\left(\sum\nolimits_{i=1}^n p_i e^{\varpi(1-p_i)}\right)}^2 } \le 0.
\end{flalign}
Since that $f'(\varpi)\le 0$ when $\varpi>0$, $D_r({\emph{\textbf{x}}},\varpi)$ is monotonically decreasing with $\varpi$ in $(0, +\infty)$. 

(2) Similarly, when $\varpi<0$, if $0<p_j<p_i$, then we will have $l_i>l_j$ due to Theorem \ref{thm:p_l_monotonicity}. Thus, $ r^{-l_i} (p_j-p_i)\ge  r^{-l_j} (p_j-p_i)$ in this case. With taking $p_i p_j  {  e^{\varpi(2-p_i-p_j)}  }\ge 0$ into account, we have
\begin{flalign} \label{equ: F1 negiative w}% \label{equ: F1_zero}
%F_1=&\sum\nolimits_{i}  \sum\nolimits_{j}   p_i p_j  {  e^{\varpi(2-p_i-p_j)}  }r^{-l_i} (p_j-p_i)        \\
%F_1&=  \sum\limits_{p_j< p_i}   p_i p_j  {  e^{\varpi(2-p_i-p_j)}  } r^{-l_i} (p_j-p_i)    + \sum\limits_{p_j> p_i}   p_i p_j  {  e^{\varpi(2-p_i-p_j)}  } r^{-l_i}  (p_j-p_i)      \\%\tag{\theequation a} \label{equ: F1_zero a} \\
F_1&   \ge      \sum\limits_{p_j< p_i}   p_i p_j  {  e^{\varpi(2-p_i-p_j)}  } r^{-l_j} (p_j-p_i)    + \sum\limits_{p_j > p_i}   p_i p_j  {  e^{\varpi(2-p_i-p_j)}  }  r^{-l_i} (p_j-p_i)    \\% \tag{\theequation b} \label{equ: F1_zero b} \\
 & =     \sum\limits_{p_j< p_i}   p_i p_j  {  e^{\varpi(2-p_i-p_j)}  } r^{-l_j} (p_j-p_i)    + \sum\limits_{ p_i>p_j}   p_i p_j  {  e^{\varpi(2-p_i-p_j)}  } r^{-l_j} (p_i-p_j)   \tag{\theequation a} \label{equ: F1 negiative w a} \\%  \tag{\theequation c} \label{equ: F1_zero c}  \\
& =     \sum\limits_{p_j< p_i}   M_{i,j} r^{-l_j} (p_j-p_i+p_i-p_j)    \tag{\theequation b}   \label{equ: F1 negiative w b}  \\
&=0,  \tag{\theequation c}   \label{equ: F1 negiative w c}%\tag{\theequation d} \label{equ: F1_zero d}
\end{flalign}
where Equation (\ref{equ: F1 negiative w a}) is obtained by exchanging the notation of subscript in the second item.

Besides, $F_2$ is still given by Equation (\ref{equ: F2_zero}), and $F_2=0$. 
As a result, $f'(\varpi) \ge 0$ when $\varpi<0$.
Therefore $D_r({\emph{\textbf{x}}},\varpi)$ is monotonically increasing with $\varpi$ in $(-\infty, 0)$.

(3) When $\varpi=0$, the storage size $l_i$ for $i=1,2,...,n$ will be all equal to $T$, and therefore $D_r({\emph{\textbf{x}}},0)=(r^{L-T}  -1)/(r^L-1) $. Based on the discussion in (1) and (2), we obtain $D_r({\emph{\textbf{x}}},\varpi)\le D_r({\emph{\textbf{x}}},0)$. The proof is completed.

%\section*{Acknowledgment}

%This work was partly supported by the China Major State Basic Research Development Program (973 Program) No. 2012CB316100(2) and National Natural Science Foundation of China (NSFC) No. 61321061.


%\end{IEEEbiography}


\end{document}